%% file: SpectraOfPerfectStateTransferHamiltoniansOnFractalLikeGraphs-revised.tex
\documentclass[12pt]{amsart}

\usepackage[all]{xy}
\usepackage[margin=0.98in, marginparwidth=1.5cm, marginparsep=0.5cm]{geometry}
\usepackage{comment}
\usepackage{verbatim}

\usepackage{tikz}
\usepackage{amssymb}
\usepackage{array}
\usepackage{xcolor}

\newcolumntype{C}[1]{>{\centering\arraybackslash}p{#1}}

\usepackage{mathrsfs}
\usepackage{physics}
\usepackage{tkz-graph}
\usepackage{amsthm}
\usepackage{amsmath,amsfonts,graphicx,mathdots,cite,wrapfig}
\usepackage[colorlinks,linkcolor=blue,citecolor=blue,pagebackref,hypertexnames=false, breaklinks]{hyperref}
\usepackage{cleveref}
\usepackage{cite}
\usepackage{mwe}
\usepackage{caption}
\usepackage{subfigure} 
\usepackage{marginnote}

\usetikzlibrary{backgrounds}

\newtheorem{theorem}{Theorem}[section]
\newtheorem{lemma}[theorem]{Lemma}
\newtheorem{proposition}[theorem]{Proposition}
\newtheorem{corollary}[theorem]{Corollary}
\newtheorem{assumption}[theorem]{Assumption}

\theoremstyle{definition}
\newtheorem{definition}[theorem]{Definition}
\newtheorem{example}[theorem]{Example}
\theoremstyle{remark}
\newtheorem{remark}[theorem]{Remark}

\numberwithin{equation}{section}

\DeclareMathOperator{\nn}{\mathbb{N}}
\DeclareMathOperator{\rr}{\mathbb{R}}
\DeclareMathOperator{\complex}{\mathbb{C}}
\DeclareMathOperator{\hamilton}{\mathbb{\textbf{H}}}
\DeclareMathOperator{\jacobi}{\textbf{J}}

\newcommand{\PI}{{\ensuremath{\mbox{\small$\Pi$\hskip.012em\llap{$\Pi$}\hskip.012em\llap{$\Pi$}\hskip.012em\llap{$\Pi$}}}}}

\begin{document}

\usetikzlibrary{decorations.pathreplacing}


\title{Spectra of Perfect State Transfer Hamiltonians on Fractal-Like Graphs}


\author{Gamal Mograby}
\address{Gamal Mograby, Mathematics Department, University of Connecticut, Storrs, CT 06269, USA}
\email{gamal.mograby@uconn.edu}

\author{Maxim Derevyagin}
\address{Maxim Derevyagin, Mathematics Department, University of Connecticut, Storrs, CT 06269, USA}
\email{maksym.derevyagin@uconn.edu}

\author{Gerald V. Dunne}
\address{Gerald V. Dunne, Mathematics \& Physics Department, University of Connecticut, Storrs, CT 06269, USA}
\email{gerald.dunne@uconn.edu}

\author{Alexander Teplyaev}
\address{Alexander Teplyaev, Mathematics \& Physics Department, University of Connecticut, Storrs, CT 06269, USA}
\email{alexander.teplyaev@uconn.edu}


\date{\today}



\begin{abstract}
In this paper we study the spectral features, on fractal-like graphs, of Hamiltonians which exhibit the special property of perfect quantum state transfer: the transmission of quantum states without dissipation. The essential goal is to develop the theoretical framework for understanding the interplay between perfect quantum state transfer, spectral properties, and the geometry of the underlying graph, in order to design novel protocols for applications in quantum information science. We present a new lifting and gluing construction, and use this to prove results concerning an inductive spectral structure, applicable to a wide variety of fractal-like graphs. We illustrate this construction with explicit examples for several classes of diamond graphs.

\end{abstract}

\maketitle

\tableofcontents

\section{Introduction}
The transfer of a quantum state from one location in a quantum network to another is a fundamental task in quantum information technologies, and such a transfer is called \textit{perfect} if it is realized with probability one, that is, without  dissipation. 
Perfect quantum state transfer (we write shortly PQST) has potential applications to the design of sub-protocols for quantum information and quantum computation \cite{Kay10,CVZh17,KLY17}. Depending on the application, various quantum systems are employed. Typical designs involve information carriers like photons in optical systems \cite{Guthhrlein2001ASI}, or phonons in ion traps \cite{Leibfried2003,SchmidtKaler2003}. Other promising devices are spin chains. The study of PQST on spin chains was pioneered by S.~Bose \cite{Bose03,Bose_review},  who  considered  a  $1D$  chain  of $N$ qubits coupled by a time-independent Hamiltonian. His work generated intense theoretical interest, in particular in questions concerning how to manipulate and engineer Hamiltonians such that a PQST is achieved. Manufacturing such manipulated Hamiltonians will provide pre-fabricated devices for quantum computer architectures, which takes input in one location and outputs it at another without needing to interact with the device. 
This approach is robust to noise and hence much less prone to errors. A number of one dimensional cases, where PQST can be achieved, have been found in some $XX$ chains with inhomogeneous couplings, see 
\cite[and references therein]{Kay10,Bose_review,christandl2004perfect,burgarth2005conclusive,burgarth2005perfect,karbach2005spin,Opperman10,Opperman12,Godsil,Godsil08,Godsil12,VZ12,qin2013high}.  
Also, it was shown that in some cases it was possible to achieve almost perfect quantum state transfer, which happens under much less restrictive conditions than a PQST \cite{BACVV10, BBVB11,VZ12c}.
Recently there has been active interest to generalize these results to graphs with potentials and to graphs that are not one dimensional \cite{P-RKay11, KLY17,pretty,KMPPZ17,KRA12,VZh12}. 
These works illustrate the fact  that PQST is a rare phenomenon, for which the construction of explicit examples remains rather non-trivial. Intending to investigate the rich interplay between quantum state transfer and geometries beyond one-dimensional graphs, we showed in a previous paper \cite{2019arXiv190908668D} that PQST is possible on the large and diverse class of fractal-type diamond graphs. A significant interest in these graphs lies in the fact that their limit spaces constitute a family of fractals, which present different geometrical properties, including a wide range of Hausdorff and spectral dimensions. These graphs have provided an important collection of structures with 
interesting physical and mathematical properties 
and a broad variety of geometries, see \cite{MT,ADT09,HK,NT,AR18,AR19,T08,MT2,brzoska2017spectra}. The structure of these graphs is such that they combine spectral properties of Dyson hierarchical models and transport properties of one dimensional chains. The methods that we use are discretized versions of the methods recently developed in 
\cite{AR18,AR19} (see also \cite{AHTT18,ST19}), which provides a construction of Green's functions for diamond fractals. 

In this paper we generalize the construction in \cite{2019arXiv190908668D} and show that it works for any graph possessing a transversal decomposition (see assumption \ref{graphAssumptions}). More precisely, on such a graph, a Hamiltonian based on nearest-neighbor coupling and with a certain transversal projective structure (see assumptions \ref{hamiltonAssumptions}) can be engineered to admit a PQST. One of the new features we present in this paper is that we can transport quantum state from multiple sites on such graphs to another set of such sites. For more details, see Theorem \ref{PQSTtheorem}.
The primary goal of this paper is to demonstrate new spectral properties of the graphs on which a PQST can be achieved. To this end we firstly need to understand the spectrum of the Hamiltonians we construct.
Advantageous settings to accomplish this task
are projective limit-type spaces. Analysis on projective limit spaces is an active area of current research \cite{Cheeger2013,Cheeger2013InverseLS}.
Barlow and Evans used projective limits to produce a new class of state spaces for Markov
processes \cite{BE}. The spectra of  Laplacians on Barlow-Evans type projective limit spaces were studied in \cite{ST19}, see also \cite{St1,St2}.
We proceed in this paper in the same spirit but dealing with Hamiltonians instead of Laplacians. To this end, we provide a discretized version of a sequence of projective limit spaces \cite[Definition 2.1, page 3]{ST19}. 
By doing so, we are able to
construct a sequence of graphs $\{G_i\}_{i \geq 0}$ and equip each
$G_i$ with a Hamiltonian $\hamilton_i$ such that PQST can be achieved (under some additional assumptions).   
Next, we provide a complete description of the spectrum of $\hamilton_i$ and for the convenience of the reader, we state the result in the following theorem, see the proof of Theorem \ref{key-thm2} for further details. 
\begin{theorem}
Let $i$ be a nonnegative integer and let $\hamilton_i$ be the Hamiltonian lifted from a Jacobi matrix $\jacobi$ on $G_0$ to $G_i$.
Then there exists a collection $\jacobi=\jacobi_0$,  $\jacobi_1$, \dots , $\jacobi_m$  of submatrices of $\jacobi$  such that 
\begin{equation*}
 \sigma(\hamilton_i) =\sigma(\jacobi) \cup\sigma(\jacobi_1)\cup \dots \cup\sigma(\jacobi_m).
\end{equation*}
 In particular,  if the Jacobi matrix $\jacobi$ corresponds to the case of PQST in the 1D chain $G_0$ then $\hamilton_i$ realizes PQST on $G_i$. Therefore, the above-given formula describes spectra of Hamiltonians realizing PQST on $G_i$.
\end{theorem} 
The Jacobi matrices
$\jacobi_0, \dots , \jacobi_m$ are easily determined by the construction scheme that generates $G_i$ from $G_{0}$. As we will see, the Jacobi matrices $\jacobi_0, \dots , \jacobi_m$ reflect geometrical information of the graph $G_i$. Moreover, this result provides a straightforward algorithm to determine the spectrum $\sigma(\hamilton_i)$.
In Section \ref{sec:twoExamples}, we demonstrate how to apply this result on two models of diamond-type graphs  and show that in this case the corresponding PSQT Hamiltonians have multiple eigenvalues, which to the best of our knowledge seems to be unnoticed until now.
These models are a particular case of the Berker lattice construction \cite{Berker_1979} and have been initially the focus of considerable work in statistical mechanics (see, for example \cite{1983JSP....33..559D,PhysRevB.28.218, Collet1985}).

Our work is part of a long term study of  mathematical physics on fractals and self-similar graphs 
\cite{be2, be1, v1, v2, ADT09, ADT10, ABDTV12, Dunne12, Akk, AkkermansDunneLevy, hanoi, ACDRT, HM19, DDMT2}, {in which novel features of quantum processes on fractals can be associated with the unusual spectral and geometric properties of fractals compared to regular graphs and smooth manifolds.}

The paper is organized as follows. 
     {Section 2 starts with the definition of a transversal layer, which is one of the fundamental concepts for our construction (see Remark \ref{TrDecInt} for some intuition). Then, we build up the technique of constructing Hamiltonians $\hamilton$ that realize PQST and as a matter of fact enlarge the class of such Hamiltonians. 
Next, Section 3 gives a partial description of spectra of the Hamiltonians $\hamilton$ by providing some generic spectral statements. After that, Section 4 defines a discrete version of a projective limit space, on which a more precise spectral description is given in Theorem \ref{key-thm2}. Section 5 demonstrates how to apply  Theorem \ref{key-thm2} to two models of diamond-type graphs. }
Section 6 discusses the results in further geometrical structures.

\section{Perfect quantum state transfer on graphs}
\label{sec: Perfect quantum state transfer on graphs}
In this section, we extend the study of PQST on diamond fractal graphs \cite{2019arXiv190908668D} to a more general class of graphs. Let $G = (V(G),E(G))$ be a finite connected graph with a vertex set $V(G)$ and an edge set $E(G)$. We equip $G$ with the geodesic metric $d: V(G) \times V(G) \to \rr$, i.e. for $x,y \in V(G)$, $ \ d(x,y)$ gives the number of edges in a shortest path connecting $x$ and $y$. Suppose $A \subset V(G)$ is a non-empty set of vertices. The distance of $A$ to a vertex $x \in V(G)$ is defined as
\begin{equation*}
    d(x;A) = min \lbrace d(x,y) : y \in A \rbrace.
\end{equation*}
The following definition generalizes the concept of the intrinsically transversal layers introduced in \cite{2019arXiv190908668D}. This concept can be found in \cite[page 76]{MR2316893} under the name stratification and plays a crucial role in the quantum decomposition of a graph adjacency matrix.
\begin{definition}
\label{def: mappingForTransversalDecomp}
Let $A \subset V(G)$, $A \neq \O$ and $n \in \nn$. An $n$-th transversal layer with respect to $A$ is defined as 
\begin{equation*}
    \PI_A^{-1}(n) = \lbrace x \in V(G) : \  d(x;A)=n \rbrace,
\end{equation*}
where $ \PI_A: V(G) \to \nn$, $ \PI_A(x)=d(x;A)$. A transversal decomposition of $G$ with respect to $A$ is defined as 
$ V(G) = \bigcup_{n}  \PI_A^{-1}(n) $. Note that $ \PI_A^{-1}(0)=A$.
\end{definition}
\begin{remark} \label{TrDecInt} The concept of transversal layer is crucial for our construction since the reference set $A$ plays the role of sender sites and the last nonempty transversal layer is the set of receiver sites. In other words, we perfectly transfer data from $A$ to the farthest layer. Moreover, in a sense each layer acts as a site in the $1D$ case modulo the lifting procedure described in this section. 
\end{remark}
     {A quantum state on $G$  is represented by a complex-valued function on the vertices $V(G)$ and such a function is also referred to as a wave function.} The following Hilbert space will be used as a domain of the constructed Hamiltonian, which realizes perfect quantum state transfer on $G$.
\begin{definition}
\label{def:HilbertOnG}
Let $A \subset V(G)$, $A \neq \O$.      {For two wave functions $\psi$ and $\varphi$ define the inner product
\begin{equation}
\label{key-inner}
	\bra{\psi}\ket{\varphi}_{A} = \sum_{x \in V} \psi(x) \overline{\varphi(x)} \mu_A(x),
\end{equation}
where the weights are given by $\mu_A(x)=\frac{1}{|\PI_A^{-1}(n)|}$ with $n = \PI_A(x)$ and $|\PI_A^{-1}(n)|$ denotes the number of vertices in the transversal layer $\PI_A^{-1}(n)$ that contains $x$. Clearly, the space $L^2(G)=\{\psi \ | \ \psi:V(G) \to \complex \}$ of quantum states on $G$ equipped with this inner product is a Hilbert space.}
\end{definition}
     {Another concept that our construction relies on is a radial function.  Within the above-given settings, a wave function is said to be radial with respect to $A$ if its values depend only on the distance from $A$.}
\begin{definition}
Let $ V(G) = \bigcup_{n=0}^N  \PI_A^{-1}(n) $  be a transversal decomposition of $G$ with respect to $A$, for some $A \subset V(G)$, $A \neq \O$. The subspace of radial functions with respect to $A$ is defined by
 \begin{equation*}
 L^2_{rad}(G)=\{\psi \in L^2(G) \ \ | \  \psi(x)=\psi(y) \text{ if } \PI_A(x)=\PI_A(y) \}.
 \end{equation*}
The projection of $L^2(G)$ onto $L^2_{rad}(G)$ is denoted by $Proj: L^2(G) \to L^2_{rad}(G)$. 
\end{definition}
The advantage of the transversal decomposition $ V(G) = \bigcup_{n=0}^N  \PI_A^{-1}(n) $ is that it induces an auxiliary 1D chain (path graph) $D_{N} = (V(D_{N}), E(D_{N}))$ with a set of vertices  $V(D_{N})=\{0,\dots, N \}$ and a set of edges $E(D_{N})=\{(n-1,n) : \ 1 \leq n \leq N \}$. A transversal layer $\PI_A^{-1}(n)$ is identified with the vertex $n$ in the sense that the vertices $n-1$ and $n$ are defined to be adjacent in the 1D chain if and only if their corresponding transversal layers are adjacent.
To reduce the perfect quantum state transfer problem from the graph $G$ to the auxiliary 1D chain $D_{N}$, we introduce the following Hilbert space $L^2(D_{N})=\{\psi \ | \ \psi: V(D_{N}) \to \complex \}$
equipped with the standard inner product
\begin{equation}
	\label{standard inner}
	\bra{\psi}\ket{\varphi} = \sum_{n=0}^N \psi(n) \overline{\varphi(n)}.
\end{equation}
Moreover we project a wave function in $L^2(G)$ to a wave function in $L^2(D_{N})$ through averaging its values on the transversal layers:
\begin{eqnarray*}
P: L^2(G)  \to  L^2(D_{N}), \quad
\psi \mapsto   P\psi(n)=\frac{1}{|\PI_A^{-1}(n)|}\sum_{x \in \PI_A^{-1}(n)}  \psi(x).
\end{eqnarray*}
\begin{lemma}
Let $P^{\ast} $ be the adjoint operator of $P$, i.e.  $\bra{P\psi}\ket{\varphi} = \bra{\psi}\ket{P^\ast \varphi}_{A}$ for $\psi \in L^2(G)$ and $\varphi \in L^2(D_{N})$. Then $P^{\ast}$ is given by
\begin{eqnarray*}
P^{\ast}: L^2(D_{N})  \to  L^2(G), \quad
\varphi  \mapsto P^{\ast} \varphi (x) = \varphi(\PI_A(x)).
\end{eqnarray*}
\end{lemma}
\begin{proof}
A simple calculation shows that
\begin{eqnarray*}
\bra{P\psi}\ket{\varphi} = \sum_{n=0}^N  P\psi(n) \overline{\varphi(n)} 
=  \sum_{x \in V} \psi(x) \overline{\varphi(\PI_A(x))} \mu_A(x) = \bra{\psi}\ket{P^\ast \varphi}_{A}.
\end{eqnarray*}
\end{proof}
We will use the following lemma later.
\begin{lemma}
\label{usefulProp1}
Let $Id_{D_{N}}: L^2(D_{N}) \to L^2(D_{N})$ be the identity operator on $L^2(D_{N})$. Then
\begin{enumerate}
	\item The range of $P^{\ast}$ is $L^2_{rad}(G)$.
	\item $ Ker P = (L^2_{rad}(G))^{\bot}$.
	\item $P P^{\ast} =Id_{D_{N}}$.
	\item $P^{\ast} P = Proj $.
\end{enumerate}
\end{lemma}
\begin{proof}
(1) and (3) follow by definition. (2) Use $ Ker P = (Range \ P^{\ast})^{\bot}$. (4) Decompose  $\psi=Proj \ \psi + \psi_{rad}^{\bot}$, i.e. $Proj \ \psi \in L^2_{rad}(G)$ and $\psi_{rad}^{\bot} \in (L^2_{rad}(G))^{\bot}$. By (2) it follows $P^{\ast} P \psi= P^{\ast} P \  Proj \ \psi =  Proj \ \psi$, where the last equality holds by the definitions of $P$ and $P^{\ast}$.
\end{proof}

     {In what follows, we will need the following mappings.}
\begin{definition}
Let $V(G)=\PI_A^{-1}(0)\cup\PI_A^{-1}(1) \ldots \cup \PI_A^{-1}(N)$ be a transversal decomposition of $G$ with respect to $A$ for some $A \subset V(G)$, $A \neq \O$ and $N \in \nn$.
We define the following mappings:
\begin{enumerate}
    \item The left-hand side degree of a vertex $\mathbf{deg}_{-}$:
    \begin{equation*}
        \mathbf{deg}_{-}: \PI_A^{-1}(1) \ldots \cup \PI_A^{-1}(N) \to \nn.
    \end{equation*}
     Let $x \in \PI_A^{-1}(n)$ for some $n \in \{1,\ldots, N \}$. The mapping $\mathbf{deg}_{-}(x)$ assigns the vertex $x$ the number of edges that connect $x$ to vertices in $\PI_A^{-1}(n-1)$.
    \item The right-hand side degree of a vertex $\mathbf{deg}_{+}$:
    \begin{equation*}
        \mathbf{deg}_{+}: \PI_A^{-1}(0) \ldots \cup \PI_A^{-1}(N-1) \to \nn.
    \end{equation*}    
    Let $x \in \PI_A^{-1}(n)$ for some $n \in \{0,\ldots, N-1 \}$. The mapping $\mathbf{deg}_{+}(x)$  assigns the vertex $x$ the number of edges that connect $x$ to vertices in  $\PI_A^{-1}(n+1)$.
        \item  The same transversal layer degree of a vertex $\mathbf{deg}_{0}$:
    \begin{equation*}
        \mathbf{deg}_{0}: \PI_A^{-1}(0) \ldots \cup \PI_A^{-1}(N) \to \nn.
    \end{equation*}
     Let $x \in \PI_A^{-1}(n)$ for some $n \in \{0,\ldots, N \}$. The mapping $\mathbf{deg}_{0}(x)$ assigns the vertex $x$ the number of edges that connect $x$ to vertices in the same transversal layer $\PI_A^{-1}(n)$.
\end{enumerate}
\end{definition}

A Hamiltonian on $G$ is a self-adjoint operator 
$\hamilton$ acting on $L^2(G)$. 
It was observed in \cite{2019arXiv190908668D} that constructing a Hamiltonian, which is not only adapted to the graph structure but also to the given transversal decomposition of the diamond-type graphs, leads indeed to a Hamiltonian that realizes a perfect quantum state transfer. Motivated by these observations, we impose the following assumptions on $\hamilton$:
\begin{assumption}[Assumptions on the Hamiltonian]
\label{hamiltonAssumptions}
The self-adjoint operator 
$\hamilton$ acting on $L^2(G)$ is assumed to satisfy the following properties:
\begin{enumerate}
\item {{\it Nearest-neighbor coupling}:} for $x,y \in V(G)$, let $\bra{x}\hamilton\ket{ y}_{A}=0$
if $x$ and $y$ are not connected by an edge, i.e., the transition matrix element from the quantum state $\ket{y}$ to $\ket{x}$ is zero if the vertices $y$ and $x$ are not adjacent in $G$.
\item 
{{\it Radial coupling}:} for $x_1,y_1,x_2,y_2 \in V(G)$ such that both $x_1,y_1$ and $x_2,y_2 $ are adjacent, we set
\begin{align*}
\bra{x_1}\hamilton\ket{y_1}_{A}=\bra{x_2}\hamilton\ket{y_2}_{A} \\  \text{if   } \ \PI_A(x_1)=\PI_A(x_2)  \  \text{ and } \  \PI_A(y_1)=\PI_A(y_2),
\end{align*}
i.e.,   the transition matrix elements are compatible with the transversal decomposition of $F$.
\item For $x,y \in \PI^{-1}_A(n)$,  $n \in \{0, \dots , N \}$ we assume 
 $\bra{x}\hamilton\ket{ x}_{A}=\bra{y}\hamilton\ket{ y}_{A}$. Moreover, if $x,y \in \PI^{-1}_A(n)$ are adjacent, then we assume 
 $\bra{x}\hamilton\ket{ y}_{A}=\bra{x}\hamilton\ket{ x}_{A}$.
\end{enumerate}
\end{assumption}
\begin{remark}
For a vertex $x \in V(G)$, the quantum state $\ket{x}$ corresponds to the one-excitation state at the vertex $x$, i.e.
\begin{equation*}
\ket{x} =
  \begin{cases}
    1       & \quad \text{on vertex } x\\
    0  & \quad \text{on } V(G) \backslash \{x\} .
  \end{cases}
\end{equation*}
\end{remark}
A Hamiltonian $\hamilton$ on $G$ is related to an operator on the 1D chain $D_{N}$ by
\begin{equation}
\label{inducedJ}
\mathbf{J}  =  P \hamilton P^{\ast},
\end{equation}
which acts on $L^2(D_{N})$. Similarly, we denote the one-excitation states in $L^2(D_{N})$ by $\ket{n} = (0,\dots, 1, \dots, 0)$ where the $1$ occupies the $n$-th position. The following proposition gives a simple criterion for determining whether the constructed Hamiltonian $\hamilton$ is self-adjoint or not.
\begin{proposition}
Let $ V(G) = \bigcup_{n=0}^N  \PI_A^{-1}(n) $  be a transversal decomposition of $G$ with respect to $A$, for some $A \subset V(G)$, $A \neq \O$ and let assumptions \ref{hamiltonAssumptions} hold. Then,
the Hamiltonian $\hamilton$ is self-adjoint with respect to the inner product \eqref{key-inner} if and only if $\mathbf{J}$ is self-adjoint with respect to the inner product \eqref{standard inner}.
\end{proposition}
\begin{proof}
Note that equation (\ref{inducedJ}) implies that $\mathbf{J}$ satisfies the nearest-neighbor coupling condition. Hence it is sufficient to consider adjacent vertices, $x \in \PI^{-1}_A(n)=\{x_{1}, \ldots, x_{k} \}$ and $y \in \PI^{-1}_A(n+1)=\{y_{1}, \ldots, y_{m} \}$ for some $n \in \{0, \dots , N-1 \}$. We observe
\begin{eqnarray*}
\bra{n}\ket{\mathbf{J}(n+1)}&=&  (\bra{x_1}+\dots + \bra{x_k})( \ket{\hamilton y_1} +\dots +  \ket{\hamilton y_m}) \\
&=& \sum_{x_{i} \in \PI^{-1}_A(n)} \mathbf{deg}_{+}(x_i) \bra{x}\ket{\hamilton y}_{A},
\end{eqnarray*}
where the second equality holds by the radial coupling assumption. Similarly,
\begin{eqnarray*}
\bra{\mathbf{J}n}\ket{n+1}&=& \sum_{y_{i} \in \PI^{-1}_A(n+1)} \mathbf{deg}_{-}(y_i) \bra{\hamilton x}\ket{ y}_{A}.
\end{eqnarray*}
The statement follows as the matching identity 
\[
\sum_{x_{i} \in \PI^{-1}_A(n)} \mathbf{deg}_{+}(x_i) = \sum_{y_{i} \in \PI^{-1}_A(n+1)} \mathbf{deg}_{-}(y_i)
\] 
holds. It gives, in fact, the number of edges between the transversal layers $\PI^{-1}_A(n)$ and $\PI^{-1}_A(n+1)$. 
We consider now the diagonal elements
\begin{eqnarray*}
\bra{n}\ket{\mathbf{J}n}&=&  (\bra{x_1}+\dots + \bra{x_k})( \ket{\hamilton x_1} +\dots +  \ket{\hamilton x_m}) \\
&=&  \sum_{x_{i} \in \PI^{-1}_A(n)} (\mathbf{deg}_{0}(x_i)+1) \bra{x}\ket{\hamilton x}_{A}.
\end{eqnarray*}
Similarly, we have $
\bra{\mathbf{J} n}\ket{n}
=  \sum_{x_{i} \in \PI^{-1}_A(n)} (\mathbf{deg}_{0}(x_i)+1) \bra{\hamilton x}\ket{ x}_{A}
$.
\end{proof}



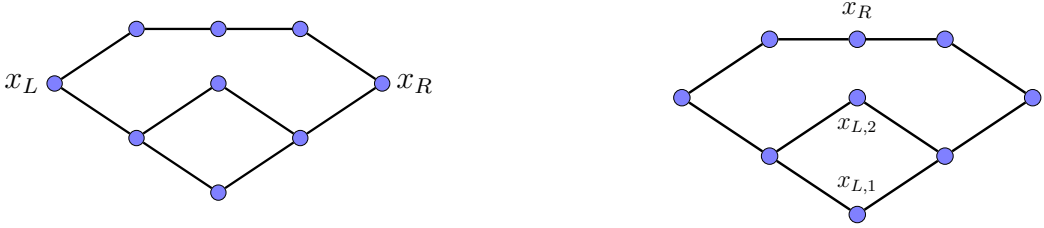
\begin{figure}[htbp]
\begin{tabular}{C{.49\textwidth}C{.49\textwidth}} 
\subfigure{
\resizebox{6.cm}{!}{\input{ExampleFrom24Dec.tikz}}

} &
\subfigure{
\resizebox{5cm}{!}{\input{ExampleFrom24Dec2.tikz}}
  
}
\end{tabular}
\caption{ (Left) This graph doesn't satisfy the graph assumptions \ref{graphAssumptions}  with respect to $A=\{x_L\}$. However, it is possible to construct a Hamiltonian that admits PQST from $A=\{x_L \}$ to $B=\{x_R\}$.  (Right) The same graph satisfies the graph assumptions \ref{graphAssumptions}  with respect to $A=\{x_{L,1},x_{L,2}\}$. Theorem \ref{PQSTtheorem} implies the possibility of constructing a Hamiltonian that admits PQST from $A=\{x_{L,1},x_{L,2}\}$ to $B=\{x_R\}$. }
\label{fig:ExampleFrom24Dec.tikz}
\end{figure}

From now on, we require that the graph $G$ satisfies the following assumption.
\begin{assumption}[Assumptions on the graph $G$]
\label{graphAssumptions} Let $G$ be a finite connected graph. We assume there exists $A \subset V(G)$, $A \neq \O$ that transversally decomposes $ V(G) = \bigcup_{n=0}^N  \PI_A^{-1}(n) $ in such a way that the following holds:
\begin{enumerate}
\item The mappings $\mathbf{deg}_{+}$, $\mathbf{deg}_{-}$ and $\mathbf{deg}_{0}$ are constant on a transversal layer, i.e., for $x,y \in \PI^{-1}_A(n)$ we have
\begin{equation*}
\mathbf{deg}_{+}(x)=\mathbf{deg}_{+}(y), \quad \mathbf{deg}_{-}(x)=\mathbf{deg}_{-}(y), \quad \mathbf{deg}_{0}(x)=\mathbf{deg}_{0}(y).
\end{equation*}
\end{enumerate}
\end{assumption}
The following lemma  follows in exactly the same way as \cite[Lemma 2, page 8]{2019arXiv190908668D}.
\begin{lemma}
\label{invariantSubspaceLemma}
Under the assumptions \ref{hamiltonAssumptions} and \ref{graphAssumptions}, we can prove that the subspace $L^2_{rad}(G)$ is invariant under $\hamilton$. 
\end{lemma}
Recall that our primary motivation is to understand      {how quantum systems beyond a $1D$} chain can be engineered to produce sub-protocols of perfect quantum state transfer.
Let $G$ be a graph transversally decomposed with respect to $A$, satisfying the assumptions \ref{graphAssumptions} and associated with the $1D$ chain $D_{N}$. We set $A=\{x_{L,1}, \dots, x_{L,m} \}$ and define the quantum state $\ket{A}=\ket{x_{L,1}}+\dots \ket{x_{L,m}}$. Note that $\ket{A}= P^{\ast}\ket{0} \in L^2_{rad}(G)$. 
Similarly, we define the quantum state  $\ket{B}= P^{\ast}\ket{N} \in L^2_{rad}(G)$. The following theorem provides a sufficient condition of how to design a Hamiltonian on $G$ that achieves a perfect transfer of the quantum state $\ket{A}$ into $\ket{B}$.  
\begin{theorem} 
\label{PQSTtheorem}
Under the assumptions \ref{hamiltonAssumptions} and \ref{graphAssumptions},
if a PQST on the 1D chain $D_{N}$ is achieved, i.e.,  there exists $T>0$ such that $e^{iT \mathbf{J}} \ket{0}=e^{i\phi}\ket{N}$ for some phase $\phi$, then a PQST on $G$ is also achieved with the same time $T$ and phase $\phi$, i.e., 
\begin{equation*}
e^{i T \hamilton } \ket{A}= e^{i\phi}\ket{
B}
\text{ \ and \ }
e^{i T \hamilton } \ket{B}= e^{i\phi}\ket{A},
\end{equation*}
where $A=\{x_{L,1}, \dots, x_{L,m}\}$ is the set of sender sites and $B$ is the set of receiver sites.  
\end{theorem}
 \begin{proof}
  In the same way as \cite[Proof of Theorem 1, page 9]{2019arXiv190908668D}, we show $e^{iT \hamilton \mathbf{Proj}} \ket{A}- e^{i\phi}\ket{B} \in Ker(P)$. Using $\ket{A}, \ket{B} \in L^2_{rad}(G)$ we conclude with Lemma \ref{invariantSubspaceLemma}, $ \ e^{iT \hamilton \mathbf{Proj}} \ket{A}= e^{i\phi}\ket{B}$. Let $\mathbf{Proj}^{\bot}$ be the projection of $L^2(G)$ onto $(L^2_{rad}(G))^{\bot}$. The statement follows by
 	$(\hamilton \mathbf{Proj} + \hamilton \mathbf{Proj}^{\bot})  \ket{A} = \hamilton \mathbf{Proj} \ket{A}$ 
 	
 \end{proof}
\begin{remark}
In a previous paper \cite{2019arXiv190908668D}, we considered the PQST from an excited state on a single vertex $x_L$ to another  excited state on a single vertex $x_R$. Theorem \ref{PQSTtheorem} covers additional situations, in which a PQST is achieved between      {the transversal layers $A\subset V(G)$ and $B\subset V(G)$, and each of those layers may} contain more than a single vertex, see Figure \ref{fig:ExampleFrom24Dec.tikz} (right). On the other hand, Figure \ref{fig:ExampleFrom24Dec.tikz} (left)  shows an example of a graph that doesn't satisfy the assumptions \ref{graphAssumptions}  with respect to $A=\{x_L\}$. However, it is possible to construct a Hamiltonian that admits a PQST from $A=\{x_L \}$ to $B=\{x_R\}$.
\end{remark}
Let $(H(x,y))_{x,y \in V(G)}$ be the matrix representation of $\hamilton$ with respect to the canonical basis $\{ \ket{x} \}_{x \in V(G)}$. The following result relates the matrix elements of $\hamilton$ to $\jacobi$ and can be proved similarly to  \cite[Proposition $1$]{2019arXiv190908668D}.  
\begin{proposition} 
\label{hamiltonianMatrix}     {
Let $x\in V(G)$. Also, let $y\in V(G)$ be adjacent to $x$ and $\PI_A(y)=\PI_A(x)\pm 1$. Then
\begin{enumerate}
	\item $H(x,x) = \dfrac{1}{\mathbf{deg}_{0}(x)+1}\bra{\PI_A(x)}\jacobi \ket{\PI_A(x)}$.
	\item $H(x,y) = \dfrac{1}{\mathbf{deg}_{\pm}(x) }\bra{\PI_A(x)}\jacobi \ket{\PI_A(x)\pm 1}$.
\end{enumerate}}
\end{proposition}

\section{Generic spectral properties of $\hamilton$}

\subsection{Radial eigenvectors of $\hamilton$}

The goal of this section is to give a partial description of the spectrum of a Hamiltonian $\hamilton$ satisfying the assumptions \ref{hamiltonAssumptions}. This part of the spectrum is related to the transversal decomposition of $G$ and consequently can be described for a generic $G$ satisfying the assumptions \ref{graphAssumptions}. The following lemmas reveal some  advantages for considering the induced $1D$ chain and the Jacobi matrix $\jacobi$ while investigating the Hamiltonian $\hamilton$. In section \ref{sec:Main inductive result}, we will see that this approach is very fruitful. In fact, we will develop this approach further to give a complete description of the spectrum $\sigma (\hamilton)$ on a broad class of graphs.
\begin{lemma}
\label{lemmaLifteigenvector}
Let $\jacobi  =  P \hamilton P^{\ast}$. Then $\sigma(\jacobi) \subset \sigma(\hamilton)$. Moreover, if $\lambda \in \sigma(\jacobi)$ is an eigenvalue with the eigenvector $v_{\lambda}$ then  $ P^{\ast} v_{\lambda}$ is a corresponding $\hamilton$-eigenvector.
\end{lemma}
\begin{proof} 
Let $\lambda \in \sigma(\jacobi)$ be an eigenvalue corresponding to the eigenvector $v_{\lambda} \in L^2(D_{N})$. Then 
\begin{equation*}
\lambda P^{\ast} v_{\lambda} = P^{\ast}\jacobi v_{\lambda} =  P^{\ast} P \hamilton P^{\ast} v_{\lambda} = Proj \ \hamilton P^{\ast} v_{\lambda}=  \hamilton P^{\ast} v_{\lambda},
\end{equation*}
where the last equality holds as $\hamilton P^{\ast} v_{\lambda} \in L^2_{rad}(G)$.
\end{proof}
Note that $P^{\ast} v_{\lambda} \in L^2_{rad}(G)$ and hence we denote it as a radial eigenvector.
\begin{lemma}
\label{resolventsRelation}
Let $z \notin \sigma(\hamilton)$. Then the resolvent operators satisfy $(\jacobi-z)^{-1}=P (\hamilton-z)^{-1} P^{\ast}$.
\end{lemma}
\begin{proof}
Note $z \notin \sigma(\hamilton)$ implies $z \notin \sigma(\jacobi)$ by Lemma \ref{lemmaLifteigenvector}. We prove that $P (\hamilton-z)^{-1} P^{\ast}$ is the inverse operator of $\jacobi-z$. We have
\begin{eqnarray*}
(\textbf{J}-z)P (\hamilton-z)^{-1} P^{\ast} &=& P (\hamilton-z) P^{\ast} P (\hamilton-z)^{-1} P^{\ast}\\
&=& P (\hamilton-z)  Proj (\hamilton-z)^{-1} P^{\ast} \\
&=& Id_{D_{N}}
\end{eqnarray*}
where the equalities hold by Lemmas \ref{usefulProp1} and \ref{invariantSubspaceLemma}. A similar argument shows that 
 $P (\hamilton-z)^{-1} P^{\ast}$ is also a left inverse of $\jacobi-z$.
\end{proof}
Let $P_{\jacobi, \lambda}$ and $P_{\hamilton, \lambda}$ be the eigenprojections corresponding to $\lambda \in \sigma(\jacobi)$ and $\lambda \in \sigma(\hamilton)$, respectively. 
\begin{theorem}\label{key-thm1}
Let $\lambda \in \sigma(\jacobi)$. Then $P_{\jacobi, \lambda} =P \  P_{\hamilton, \lambda}P^{\ast}$.
\end{theorem}
\begin{proof}
The spectral representation of the resolvent operators in Lemma \ref{resolventsRelation} gives
\begin{equation}
\label{specRepOfResol}
\sum_{\tilde{\lambda} \in \sigma(\jacobi)}\frac{1}{z - \tilde{\lambda}}P_{\jacobi, \tilde{\lambda}}=\sum_{\tilde{\lambda} \in \sigma(\hamilton)}\frac{1}{z - \tilde{\lambda}}P P_{\hamilton, \tilde{\lambda}}P^{\ast}
\end{equation}
Multiplying both sides of equation (\ref{specRepOfResol}) by $z-\lambda$ and subsequently taking the limit $z \to \lambda$ yields the result.
\end{proof}
For the rest of the paper, we assume that the auxiliary $1D$ chain is equipped with the following Jacobi matrix
\begin{equation}
\label{jacobiMatrix}
\jacobi =
 \begin{pmatrix}
  B_{1} & J_1 & 0 &  & \\
  J_{1} & B_{2} & J_2 & 0 &  \\
    0 &  J_2 & B_2 & \ddots &    \\
   & 0  &  \ddots &    \ddots& J_{N}  \\
    & &  & J_{N} & B_{N+1}
\end{pmatrix},
\end{equation}
where $B_1, \dots , B_{N+1} \in \rr$ and $J_i > 0$ for $i \in \{1, \dots, N\}$. Let $\{ p_0(z), \dots, p_{N+1}(z) \}$ be monic polynomials defined by the 
 recurrence relations:
\begin{equation}
\label{recurrence relations}
      \begin{cases}
     p_0(z) =1, \quad p_1(z)= z- B_1, \\
     p_k(z) = (z - B_k)p_{k-1}(z) - J_{k-1}^2 p_{k-2}(z), \quad k=2,3, \dots, N+1 
  \end{cases}
  \end{equation}
The following proposition summarizes some useful spectral properties of $\jacobi$ (for more details, see \cite[p.~48]{MR2316893}).
\begin{proposition}
Every zero of $p_{N+1}(z)$ is real and simple. Moreover, $\sigma(\textbf{J}) = \{\lambda \in \complex : p_{N+1}(\lambda)=0  \}$. For an eigenvalue $\lambda \in \sigma(\textbf{J})$, the corresponding eigenvectors is given by
\begin{equation}
\label{jacobiEigenvector}
v_{\lambda} = \left(p_0(\lambda), \frac{p_1(\lambda)}{J_1}, \dots, \frac{p_{N}(\lambda)}{J_1 \cdots J_{N}} \right)^t
\end{equation}
\label{propRadialeigenvectors}
\end{proposition}

\begin{corollary}
\label{coroRadialeigenvectors}
Let $\lambda  \in  \sigma(\textbf{J})$. Then the corresponding $\hamilton$-eigenvector is given $P^{\ast}v_{\lambda}$, where $v_{\lambda}$  is defined in (\ref{jacobiEigenvector}). 
\end{corollary}

\subsection{Lifting-\&-Gluing Lemma}
\begin{figure}[htb]
\centering
\resizebox{13cm}{!}{\input{liftBranchingLemma.tikz}}
	\caption{ (Left) Three copies of the $1D$ chain $D_4$ with $V(D_4) = \{0, \dots , 4 \}$. The $i$-th copy is denoted by $D_4 \times \{ w_i\}$, where $w_i$ is a letter in the alphabet $W=\{w_1,w_2,w_3\}$. (Right) The graph $G_D$ is constructed by gluing the three copies at the boundary points. Another way of saying this is that the graph $G_D$ is made up of three branches, the $w_1$-branch, $w_2$-branch and $w_3$-branch.}
	\label{fig:liftBranchingLemma1}
\end{figure}
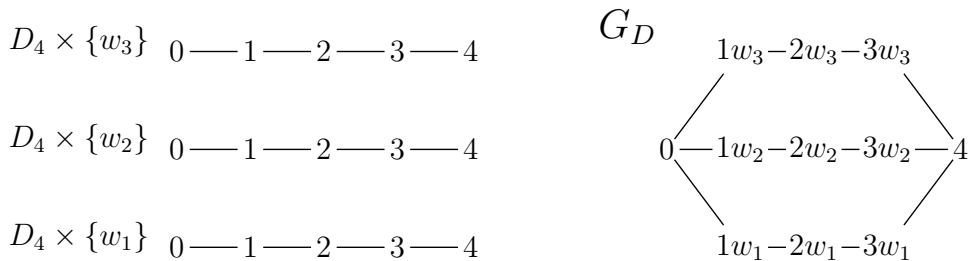
     {In this section we will prove a lemma that is essential for the reminder of the paper.} We consider a $1D$ chain $D_{N} $ equipped with  a Jacobi matrix $\jacobi$. When Dirichlet boundary conditions are imposed, we write $\jacobi^D$ for the Jacobi matrix. For a given $k \in \nn$, $k \geq 2$ we define $G_{D}$ to be the graph that is constructed by taking $k$ copies of $D_{N}$ and gluing their boundary vertices together as shown in Figure
\ref{fig:liftBranchingLemma1}. To distinguish between the copies, we use the following notation: given a $k$-letter alphabet $\{w_1, \dots, w_k\}$, we denote the $i$-th copy of $D_N$ 
by $D_N \times \{ w_i\}$ and refer to the associated subgraph in $G_D$ as the \textit{$w_i$-branch} of $G_D$.
The graph $G_D$ satisfies the assumptions \ref{graphAssumptions} with respect  to $A=\{0\}$
and $D_N$ is the auxiliary $1D$ chain. 
     {The following result is one of the key ingredients in our construction.}
\begin{lemma}[Lifting-\&-Gluing Lemma]
Let $\lambda \in \sigma(\jacobi^D)$ and $v^D_{\lambda}$ be the corresponding $\jacobi^D$-eigenvector. We define $v_{\lambda}$ to be the vector on $G_D$ that coincides with $v^D_{\lambda} $ on a $w_i$-branch and coincides with $-v^D_{\lambda} $ on another branch, say $w_j$-branch, for some $j \neq i$, i.e.,
\begin{equation}
    \label{eigenvectExtension}
    v_{\lambda} =
  \begin{cases}
    v^D_{\lambda}       & \quad \quad \text{on the $w_i$-branch  } \\
    -v^D_{\lambda}  & \quad \quad \text{on the $w_j$-branch } 
    \\
    0       & \quad \quad \text{elsewhere } 
  \end{cases}
\end{equation}
Then $\lambda \in \sigma(\hamilton)$ and $v_{\lambda}$ is an  $\hamilton$-eigenvector corresponding to $\lambda$. Moreover, $v_{\lambda} \in (L^2_{rad}(G_D))^{\bot}$.
\label{lemma: mainLemmaLifttoBranches}
\end{lemma}
In other words, if we lift a  $\jacobi^D$-eigenvector (Dirichlet eigenvector of $\jacobi$) to a branch and lift the same vector with the opposite sign to another branch, and if after that we  assign zero to the remaining branches and glue them together, then this will result in an eigenvector of $\hamilton$ on $G_D$ corresponding to the same eigenvalue. An immediate consequence of Lemma \ref{lemma: mainLemmaLifttoBranches} is that the spectrum of $\hamilton$
is determined by the spectra of $\jacobi$ and $\jacobi^D$.
\begin{corollary}
\label{cor: LiftingAndBranchingCorollary}
$\sigma(\hamilton)=\sigma(\jacobi)\cup \sigma(\jacobi^D)$
\end{corollary}
\begin{proof}
The radial eigenvectors are constructed according to Lemma \ref{lemmaLifteigenvector}, which implies $\sigma(\jacobi) \subset \sigma(\hamilton)$. The remaining eigenvectors are elements of $ (L^2_{rad}(G_D))^{\bot}$ and constructed by the Lifting-\&-Gluing Lemma \ref{lemma: mainLemmaLifttoBranches}. 
Note that for the $1D$ chain  $D_{N}$
the Jacobi matrices $\jacobi$ and $\jacobi^D$ have $N+1$ and $N-1$ eigenvectors, respectively. Each $\jacobi$-eigenvector is lifted to a radial $\hamilton$-eigenvector on $G_D$ and each $\jacobi^D$-eigenvector generates $k-1$ different $\hamilton$-eigenvectors on $G_D$. Note that the graph $G_D$ has $(N+1)+(N-1)(k-1)$ vertices.
\end{proof}
The observation in Corollary \ref{cor: LiftingAndBranchingCorollary}
is the first step in the approach that will be further developed in the next section. Indeed, we are going to show that the $\hamilton$-spectra on a broad class of graphs are determined by the spectra of a collection of Jacobi matrices.

\section{Projective limit constructions}
The following definitions are roughly speaking a discrete version of \cite[Definition 2.1, page 3]{ST19}.
\begin{definition}
Let $k \geq 2$. We refer to a $k$-letter alphabet $\{w_1, \dots, w_k\}$ as a \textit{vertical multiplier space}. A word of length $m$ is an element of the $m$-fold product $W^m = W_1 \times \dots \times W_m$ for some vertical multiplier spaces $W_1, \dots, W_m$. 
For a word $w \in W^m$, we write shortly $w = w_1 \dots w_m$ instead of $w = (w_1, \dots, w_m)$.
\end{definition}
Note that the vertical multiplier spaces $W_1, \dots, W_m$ are not assumed to have the same number of letters.
\begin{definition}
\label{def:horizontalBaseSpace}
We initialize the graph $G_0 = (V(G_0), E(G_0))$ to be a  $1D$ chain $D_N$ for some $N \geq 1$. We call $G_0$ the \textit{horizontal base space}.
\end{definition}
\begin{remark}
The assumptions on the horizontal base space in \cite{ST19} are very general (local compact second countable Hausdorff space). In this sense, Definition \ref{def:horizontalBaseSpace} represents a discretization of a specific case.
\end{remark}
\begin{definition}[cf. \cite{MR3535868}]
\label{projeLimConDef}
Given a sequence of vertical multiplier spaces $\{W_i \}_{i \geq 1}$ and a horizontal base space
$G_0=D_N$. We define a sequence of graphs $\{ G_i \}_{i \geq 0}$ inductively.
\begin{enumerate}
\item Suppose $G_{i-1}=(V(G_{i-1}),E(G_{i-1}))$ is given for some integer $i \geq 1$.
\item Choose a subgraph $B_i = (V(B_i),E(B_i))$ of $G_{i-1}$, such that $G_{i-1} \backslash B_i$ is a collection of $1D$ chains. Note $B_i$ may be an  edgeless or a disconnected subgraph. 
\item For a $1D$ chain $D$ in $G_{i-1} \backslash B_i$, we set $G_D$ to be the graph that is constructed by taking the copies $D \times \{ w_k\}$ for $w_k \in W_i$ and gluing their boundary vertices together as shown in Figure \ref{fig:liftBranchingLemma1}.
\item We construct $G_{i}$ by replacing each $1D$ chain $D$ in $G_{i-1} \backslash B_i$ with the corresponding $G_D$. 
\end{enumerate}
For convenience, we set $ V(G_i) = [(V(G_{i-1}) \setminus V(B_i)) \times W_i ]\bigcup V(B_i) $ for the set of vertices of $G_i$ and $ E(G_i) = [(E(G_{i-1}) \setminus E(B_i)) \times W_i ]\bigcup E(B_i) $ for the set of edges of $G_i$, see Figure \ref{fig:hamblyKumagaiIntermediateStep} (Left).
\end{definition}
\begin{figure}[htbp]
\begin{tabular}{C{.49\textwidth}C{.49\textwidth}} 
\subfigure{
  \resizebox{!}{!}{
  \input{two1DchainsToG1.tikz}

    }
} &
\subfigure{
\input{mappingDiagFromBen.tikz}
  
}
\end{tabular}
\caption{ (Left) To construct the graph $G_1$, we initialize the horizontal base space $G_0$ to be the $1D$ chain $D_4$ with the vertices $\{0,\dots, 4\}$.      {We set $W=\{w_1,w_2\}$ to be the vertical multiplier space.} Then $G_1$ is constructed as in Definition \ref{projeLimConDef}, where we choose the subgraph $B_1$ such that $V(B_1)=\{0,4\}$ and $E(B_1)= \emptyset$. Note that the address assignments of the vertices described in Definition \ref{projeLimConDef} are shown on the graph of $G_1$. (Right) A diagram shows how the different mappings from Definition \ref{def:projectivelimitMappings} are related to each other. }
\label{fig:hamblyKumagaiIntermediateStep}
\end{figure}
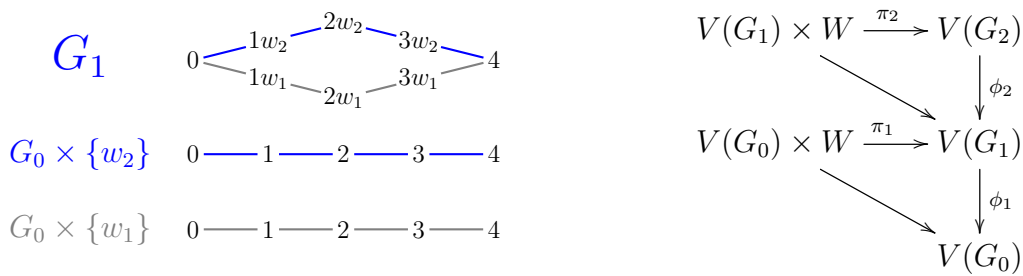
\begin{definition}
\label{def:projectivelimitMappings}
Let $\{ G_i \}_{i \geq 0}$ be constructed as described in Definition \ref{projeLimConDef}. We define $\pi_{i}: V(G_{i-1})  \times W_i \longrightarrow V(G_{i})$ by
\begin{eqnarray*}
\pi_{i}(x,w) = \left\{ \begin{array}{ll}(x,w) & \text{ if } \ x \in V(G_{i-1}) \setminus V(B_i) \\ x & \text{ if } \ x \in V(B_i), \end{array}\right.
\end{eqnarray*}
and       {also define the mapping} $\phi_{i} : V(G_i) \rightarrow V(G_{i-1})$  by 
$$\left. \begin{array}{rl}\phi_{i}(x,g) =x &  \text{ if } \ \ x \in V(G_{i-1}) \setminus V(B_i) \\ \phi_{i}(x) =x &  \text{ if } \ \ x \in V(B_i). \end{array}\right.$$
\end{definition}
The following proposition shows that each graph in $\{ G_i \}_{i \geq 0}$ admits a natural transversal decomposition, where the horizontal base space $G_0$ is used as the common auxiliary $1D$ chain for the entire sequence $\{ G_i \}_{i \geq 0}$.
\begin{proposition}
\label{prop: assumptionsForGi}
 Let $\{ G_i \}_{i \geq 0}$ be constructed as described in Definition \ref{projeLimConDef}. Then for each $i \geq 1$, the graph $G_i$ can be transversally decomposed with respect to  $A_i = (\phi_i)^{-1} \circ \dots \circ (\phi_1)^{-1}(0) \subset V(G_i)$ so that \cref{graphAssumptions} holds.
Moreover, for $ \PI_{A_i}(x)=d(x;A_i)$, see Definition \ref{def: mappingForTransversalDecomp}), we have
$\PI_{A_i}^{-1}(x)=\phi_1 \circ \dots \circ \phi_i(x)$.
\end{proposition}
\begin{proof}
Note that a vertex in $G_i$ is denoted by $nw_1w_2 \dots w_k$, where $n \in \{0,\dots, N \}$ and $w_1w_2 \dots w_k \in W^k$. The word $w_1w_2 \dots w_k$ can be considered as a vertical coordinate which gives the address of the branch that contains this vertex. On the other hand, the integer $n$ can be considered as a radial coordinate, which gives the distance to $A_i = (\phi_i)^{-1} \circ \dots \circ (\phi_1)^{-1}(0)$. By Definition \ref{def:projectivelimitMappings} we have $\phi_1 \circ \dots \circ \phi_i(nw_1w_2 \dots w_k)=n$ and therefore, this implies $ \PI_{A_i}(x)=\phi_1 \circ \dots \circ \phi_i(x)$. Now
$G_0$ as a $1D$ chain, it admits a trivial transversal decomposition with respect to $\{0\}$ i.e., $V(G_0)=\PI_0^{-1}(0)\cup\PI_0^{-1}(1) \ldots \cup \PI_0^{-1}(N), \quad \PI_0^{-1}(n)=\{n\}$.
Similarly, $G_i$ admits a transversal decomposition with respect to $A_i$,
\begin{equation*}
V(G_i)=\PI_{A_i}^{-1}(0)\cup\PI_{A_i}^{-1}(1) \ldots \cup \PI_{A_i}^{-1}(N),
\end{equation*}
where for $x \in V(G_i)$, we have
$x \in \PI_{A_i}^{-1}(n) \ \Longleftrightarrow \  \phi_{1} \circ \dots \circ \phi_{i}(x) \in \PI_0^{-1}(n)$.
\end{proof}
Each graph in $\{ G_i \}_{i \geq 0}$ admits a natural transversal decomposition. One may wonder if these graphs also satisfy the graph assumptions \ref{graphAssumptions} with respect to this decomposition. The following example shows that this is not true in general.
\begin{example}
Let the graph $\tilde{G}_2$
be constructed as described in Figure \ref{fig:graphDoesNotSatAssump}. $\tilde{G}_2$ does not satisfy the graph assumptions \ref{graphAssumptions}
 as the mappings $\mathbf{deg}_{+}$ and $\mathbf{deg}_{-}$ are NOT constant on the transversal layer $\PI^{-1}_0(2)=\{2w_2w_2,2w_1,2w_2w_1\}$: 
\begin{eqnarray*}
2=\mathbf{deg}_{+}(2w_1) \neq \mathbf{deg}_{+}(2w_2w_2)=\mathbf{deg}_{+}(2w_2w_1)=1, \\ 2=\mathbf{deg}_{-}(2w_1) \neq \mathbf{deg}_{-}(2w_2w_2)=\mathbf{deg}_{-}(2w_2w_1)=1
\end{eqnarray*}
\end{example}

\begin{figure}[htbp]
  \resizebox{!}{!}{
  \input{graphDoesNotSatAssump.tikz}

    }
\caption{ The graphs $\tilde{G}_2$ and $G_2$
are constructed as described in Definition \ref{projeLimConDef}. While $G_2$ satisfies the graph assumptions \ref{graphAssumptions}, $\tilde{G}_2$ does not.
For the construction of  $\tilde{G}_2$ and $G_2$, we set $G_1$ to be the graph shown in Figure \ref{fig:hamblyKumagaiIntermediateStep}.
(Left) $\tilde{G}_2$ is constructed by taking the two copies $G_1 \times \{w_1\}$, $G_1 \times \{w_2\}$ and choosing the subgraph $\tilde{B}_2$ such that $V(\tilde{B}_2)=\{0,2w_1 \}$ and $E(B_2)= \emptyset$. (Right)  $G_2$ is constructed by taking the two copies $G_1 \times \{w_1\}$, $G_1 \times \{w_2\}$ and choosing the subgraph $B_2$ such that $V(B_2)=\{0,2w_1, 2w_2,4 \}$ and $E(B_2)= \emptyset$. Note that $G_2$ is the level-$2$ Hambly-Kumagai diamond graph and denoted by $HK_2$. For more details see \cref{rem-key} and   \cref{sec:HamblyKumagaiDiamond}.}
\label{fig:graphDoesNotSatAssump}
\end{figure}
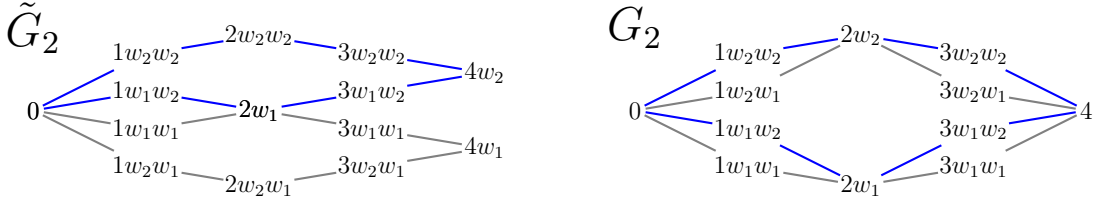


\subsection{Main Inductive Result}
\label{sec:Main inductive result}
 
      {We present a discrete version of a general construction in \cite{ST19}}. 
\begin{enumerate}
    \item We begin with a sequence of vertical multiplier spaces $\{W_i \}_{i \geq 1}$ and  a $1D$ chain $G_0$.	
    \item We construct $\{G_i\}_{i \geq 0}$ and $\{B_i\}_{ i \geq 1}$   as described in Definition \ref{projeLimConDef}.
    \item  We transversally decompose  $\{G_i\}_{i \geq 0}$ as described in Proposition \ref{prop: assumptionsForGi} and require that each $G_i$ satisfies the graph assumptions \ref{graphAssumptions} with respect to this decomposition.
\end{enumerate}

The horizontal base space $G_0$ plays the role of the auxiliary $1D$ chain and will be used to lift a Hamiltonian to each $G_i$, $i \geq 1$. To this end, we equip $G_0$ with a Jacobi matrix $\jacobi$ of the form (\ref{jacobiMatrix}). The Jacobi matrix acts on the Hilbert space $L^2(G_0)=\{\psi \ | \ \psi: V(G_0) \to \complex \}$, $\bra{\psi}\ket{\varphi} = \sum_{n =0}^N \psi(n) \overline{\varphi(n)}$. Recall that the transversal decomposition of $G_i$ is with respect to $A_i = (\phi_i)^{-1} \circ \dots \circ (\phi_1)^{-1}(0) \subset V(G_i)$. Hence, we proceed as in Section \ref{sec: Perfect quantum state transfer on graphs} and equip each $G_i$ with the Hilbert space $L^2(G_i)=\{\psi \ | \ \psi:V(G_i) \to \complex \}$, $ \bra{\psi}\ket{\varphi}_{A_i} = \sum_{x \in V(G_i)} \psi(x) \overline{\varphi(x)} \mu_{A_i}(x)$, where the weights are given by $\mu_{A_i}(x)=1/|\PI_{A_i}^{-1}(n)|$ for $n = \PI_{A_i}(x)$ and $|\PI_{A_i}^{-1}(n)|$ denotes the number of vertices in the transversal layer $\PI_{A_i}^{-1}(n)$ that contains $x$. 
     {Another useful object is the pullback operator induced by $\phi_{i}:V(G_{i})\longrightarrow V(G_{i-1})$ and it is defined as follows}
\begin{equation*}
\phi_{i}^{\ast}:L^2(G_{i-1})\longrightarrow L^2(G_i), \quad \varphi \longrightarrow \phi_{i}^{\ast}  \varphi(x)=\varphi(\phi_{i}(x)).
\end{equation*}
The  averaging operator and its adjoint are given by
\begin{eqnarray*}
P_{i}: L^2(G_i)  \to  L^2(G_0),& \quad
\psi \mapsto   P_i\psi(n)=\frac{1}{|\PI_{A_i}^{-1}(n)|}\sum_{x \in \PI_{A_i}^{-1}(n)}  \psi(x). \\
P_{i}^{\ast}: L^2(G_0)  \to  L^2(G_i),& \quad
\varphi  \mapsto P_i^{\ast} \varphi (x) =\phi_{1}^{\ast}\dots \phi_{i}^{\ast}\varphi(x)= \varphi(\phi_{1} \circ \dots  \circ \phi_{i}(x)).
\end{eqnarray*}
We are now in a position to construct 
a Hamiltonian $\hamilton_i$ on $G_i$, $i \geq 1$:
\begin{enumerate}
    \item Let $\hamilton_i$ be a Hamiltonian on $G_i$ that satisfies the assumptions \ref{hamiltonAssumptions}.
    \item Let $\hamilton_i$ be lifted from $G_0$ to $G_i$ via $\jacobi= P_{i} \hamilton_i P_{i}^{\ast}$.
\end{enumerate}
The following result is a straightforward generalization of Lemma
\ref{lemmaLifteigenvector}, see \cite{ST19}.
\begin{lemma}
\label{genlemmaLiftforInduct}
We have that $\sigma(\hamilton_{i-1}) \subset \sigma(\hamilton_{i})$. Moreover, if $\lambda \in \sigma(\hamilton_{i-1})$ is an eigenvalue corresponding to the eigenvector $v_{\lambda}$ then  $ \phi_{i}^{\ast}v_{\lambda}$ is an $\hamilton_i$-eigenvector with the same eigenvalue.
\end{lemma}
     {The following theorem characterizes the spectrum of the Hamiltonian  $\hamilton_i$ on $G_i$, $i \geq 1$.}
\begin{theorem}\label{key-thm2}
Let $i$ be a nonnegative integer and let $\hamilton_i$ be the Hamiltonian lifted from $G_0$ to $G_i$, that is, there is a Jacobi matrix $\jacobi$ such that $\jacobi= P_{i} \hamilton_i P_{i}^{\ast}$.
Then there exists a collection $\jacobi=\jacobi_0$,  $\jacobi_1$, \dots , $\jacobi_m$  of submatrices of $\jacobi$  such that 
\begin{equation*}
 \sigma(\hamilton_i) =\sigma(\jacobi) \cup\sigma(\jacobi_1)\cup \dots \cup\sigma(\jacobi_m).
\end{equation*}
 In particular,  if the Jacobi matrix $\jacobi$ corresponds to the case of PQST in a 1D chain then $\hamilton_i$ realizes PQST on $G_i$. Therefore, the above-given formula describes spectra of Hamiltonians realizing PQST on $G_i$. 
\end{theorem} 
\begin{proof}
Assume that the statement is correct for $\sigma(\hamilton_{i-1})$. By definition $G_i$ is constructed by replacing each $1D$ chain in $G_{i-1} \backslash B_i$  with multiple copies glued together at their boundary vertices. Let $\jacobi_0, \dots, \jacobi_k$ be the collection of the Jacobi matrices associated with the $1D$ chains in $G_{i-1} \backslash B_i$. Using Lemma \ref{genlemmaLiftforInduct} combined with Lemma \ref{lemma: mainLemmaLifttoBranches} (Lifting-\&-Gluing Lemma), we obtain $ \sigma(\hamilton_i) = \sigma(\jacobi_0) \cup\sigma(\jacobi_1)\cup \dots \cup\sigma(\jacobi_k)\cup \sigma(\hamilton_{i-1}) $.
\end{proof}

\begin{remark} 
     {A Jacobi matrix is one of the canonical forms of self-adjoint operators and there is a wide class of self-adjoint operators which are unitarily equivalent to the direct sum of Jacobi matrices. Therefore, such a spectral decomposition is valid for a more general class of graphs and, in fact, various approaches have been applied to obtain similar representations in different situations}, see \cite{Berkolaiko, MR3154573}. Our main assumptions on $\hamilton_i$ are to ensure that  $\hamilton_i$ realizes PQST in relatively simple practical situations. If one is interested in general Hamiltonians, the transversal decomposition of the graph and the assumptions \ref{graphAssumptions} can be relaxed for this particular statement and the proof will remain the same, as was pointed out by one of the referees. 
\end{remark} 
\begin{remark}\label{rem-key}
	Note that graphs in \cref{fig:ExampleFrom24Dec.tikz,fig:graphDoesNotSatAssump,fig:sqrt3eigenvectorsNeuDirNew,fig:plautlangToG1,fig:plautLangToG2} can be constructed in two different, but equivalent ways: either gluing copies of smaller graphs, or using inductive (projective) procedures \cite{MT,MT2,AR18,AR19,ST19}.
	
Note also that, although $\tilde{G}_2$ does not satisfy \cref{graphAssumptions}, a perfect quantum state transfer Hamiltonian can be found using the methods of \cite[Assumtions 2.8 or 2.11]{DDMT2} which is a discrete version of \cite{ST19}. 
\end{remark}


\begin{figure}[htbp]

\begin{tabular}{C{.33\textwidth}C{.33\textwidth}C{.33\textwidth}}
\subfigure{
    \resizebox{3.4cm}{!}{%
     \input{Neu_level2_eSqrt3_1st.tikz}
    }
} & 
\subfigure{
  \resizebox{3.4cm}{!}{%
    \input{Neu_level2_e0_6th.tikz}
   }
} &

\subfigure{
  \resizebox{3.5cm}{!}{
  \input{Neu_level2_e0_3rd.tikz}
    }
} 
\end{tabular}
\caption{ Hambly-Kumagai diamond graphs level 2: (Left) $\hamilton_2$-eigenvector for the eigenvalue $\sqrt{3}$. (Middle) $\hamilton_2$-eigenvector for the eigenvalue $0$. Both eigenvectors are examples for the construction method described in step 2.  (Right) $\hamilton_2$-eigenvector for the eigenvalue $0$. This eigenvector is an example for the construction method described in step 3. The number assigned to a vertex is the value of the eigenvector at this vertex.}
\label{fig:sqrt3eigenvectorsNeuDirNew}
\end{figure}
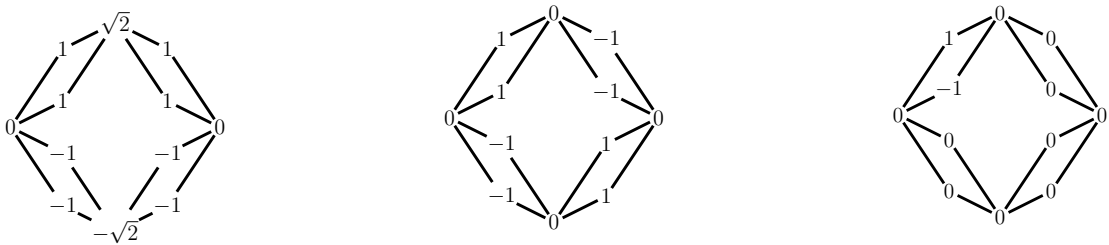


\section{Two examples}
\label{sec:twoExamples}

In this section,  we give two particular examples of the proposed construction. Note that in both cases the sender site is the leftmost point of the graphs and the receiver site is the rightmost point of the graph. We also demonstrate the applicability of Theorem \ref{key-thm2} on these two models of diamond-type graphs. A  transversal decomposition of each of these models induces a $1D$ chain $D_N$. We equip $D_N$ with      {a Jacobi matrix $\jacobi$ realizing one of the simplest cases} of spin chains with perfect state transfer discussed in \cite{christandl2004perfect}. To this end, we set
\begin{equation}
\label{KrawtchoukCoupling}
J_n=\frac{\sqrt{n(N+1-n)}}{2}, \quad B_n=0
, \quad n=0,1,\dots N, \quad B_{N+1}=0,
\end{equation}  
for the entries in (\ref{jacobiMatrix}). The underlying Jacobi matrix is mirror symmetric and it corresponds to the symmetric Krawtchouk polynomials \cite{Szego}. Following Proposition \ref{hamiltonianMatrix}, we lift this Jacobi matrix to Hamiltonians on these models of diamond-type graphs. Note that the magnetic field on the $1D$ chain nodes is assumed to vanish $B_0=\ldots = B_{N+1}=0$, resulting in a Hamiltonian whose diagonal elements are all equal to zero.  Moreover, Theorem \ref{PQSTtheorem}
implies that such Hamiltonians achieve a perfect quantum state transfer.      {We investigate these Hamiltonians and give a complete description of their spectra. }

\begin{minipage}{.49\textwidth}
\begin{equation}
    \label{eq: level2Hchain}
    \jacobi_2 =
     \left[\begin{matrix}0 & 1 & 0 & 0 & 0\\1 & 0 & \frac{\sqrt{6}}{2} & 0 & 0\\0 & \frac{\sqrt{6}}{2} & 0 & \frac{\sqrt{6}}{2} & 0\\0 & 0 & \frac{\sqrt{6}}{2} & 0 & 1\\0 & 0 & 0 & 1 & 0\end{matrix}\right]
\end{equation}
\end{minipage}
\begin{minipage}{.49\textwidth}
   \centering
     \resizebox{3.3cm}{!}{%
    \input{Neu_level2_e1_1st.tikz}
   }
   
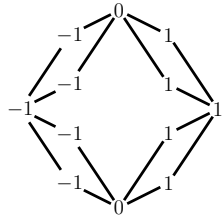
\captionof{figure}{An example of a radial eigenvector of \textbf{H$_2$}.  It corresponds to the eigenvalue 1. A number assigned to a vertex is the value of the eigenvector at this vertex.}
   \label{fig:Neu_level2_e1_1st}
\end{minipage}

\subsection{Hambly-Kumagai diamond graphs}
\label{sec:HamblyKumagaiDiamond}
The first model is an example of a two-point  self-similar graph in the sense of \cite{MT}. 
It is a particular sequence of diamond-type graphs, that was investigated in \cite{HK}. We will refer to this model as \textit{Hambly-Kumagai diamond graphs}. The following definition gives a formal description of the Hambly-Kumagai diamond graphs.
\begin{definition}
We refer to a sequence of graphs $\{HK_{\ell}\}_{\ell \geq 0}$ as \textit{Hambly-Kumagai diamond graphs}, when it is constructed as follows.
\begin{itemize}
    \item $HK_0$ is initialized as the one edge graph connecting a node $x_L$ with another node $x_R$. 
    \item At level $\ell$ we construct $HK_{\ell}$ by replacing each edge from the previous level $HK_{\ell-1}$ by two new branches, whereas each new branch is then segmented into two edges that are arranged in series.
\end{itemize}
\label{HamblyKumagaiDiamond}
\end{definition}
The first three levels of the Hambly-Kumagai diamond graphs are displayed in \cite[Figure 2, page 5]{2019arXiv190908668D}.
     {Let $V(HK_{\ell})$ be the set of vertices of $HK_{\ell}$}. It is easily seen that the transversal decomposition $V(HK_{\ell})=\PI_A^{-1}(0)\cup\PI_A^{-1}(1) \ldots \cup \PI_A^{-1}(N)$ with respect to $A=\{x_L\}$ induces a $1D$ chain $D_{N}$ such that  $N=2^\ell$. The Jacobi matrix associated with $D_{N}$, $N=2^\ell$ is denoted by  $\jacobi_{\ell}$. 
For example  $\jacobi_{2}$ is given in equation (\ref{eq: level2Hchain}). We then lift $\jacobi_{\ell}$ to a Hamiltonian $\hamilton_{\ell}$ on $HK_{\ell}$. The Hamiltonian $\hamilton_{2}$ on the Hambly-Kumagai diamond graph of level 2 is given in equation (\ref{eq:level2Hamiltonian1}).

\begin{equation}
    \label{eq:level2Hamiltonian1}
    \hamilton_2 =
     \left[\begin{array}{cccccccccccc}0 & 0 & 0 & 0 & \frac{1}{4} & \frac{1}{4} & \frac{1}{4} & \frac{1}{4} & 0 & 0 & 0 & 0\\0 & 0 & 0 & 0 & 0 & 0 & 0 & 0 & \frac{1}{4} & \frac{1}{4} & \frac{1}{4} & \frac{1}{4}\\0 & 0 & 0 & 0 & \frac{\sqrt{6}}{4} & \frac{\sqrt{6}}{4} & 0 & 0 & \frac{\sqrt{6}}{4} & \frac{\sqrt{6}}{4} & 0 & 0\\0 & 0 & 0 & 0 & 0 & 0 & \frac{\sqrt{6}}{4} & \frac{\sqrt{6}}{4} & 0 & 0 & \frac{\sqrt{6}}{4} & \frac{\sqrt{6}}{4}\\1 & 0 & \frac{\sqrt{6}}{2} & 0 & 0 & 0 & 0 & 0 & 0 & 0 & 0 & 0\\1 & 0 & \frac{\sqrt{6}}{2} & 0 & 0 & 0 & 0 & 0 & 0 & 0 & 0 & 0\\1 & 0 & 0 & \frac{\sqrt{6}}{2} & 0 & 0 & 0 & 0 & 0 & 0 & 0 & 0\\1 & 0 & 0 & \frac{\sqrt{6}}{2} & 0 & 0 & 0 & 0 & 0 & 0 & 0 & 0\\0 & 1 & \frac{\sqrt{6}}{2} & 0 & 0 & 0 & 0 & 0 & 0 & 0 & 0 & 0\\0 & 1 & \frac{\sqrt{6}}{2} & 0 & 0 & 0 & 0 & 0 & 0 & 0 & 0 & 0\\0 & 1 & 0 & \frac{\sqrt{6}}{2} & 0 & 0 & 0 & 0 & 0 & 0 & 0 & 0\\0 & 1 & 0 & \frac{\sqrt{6}}{2} & 0 & 0 & 0 & 0 & 0 & 0 & 0 & 0\end{array}\right]
\end{equation}

\subsubsection{Spectrum of the Hamiltonian $\hamilton_{2}$}
\begin{figure}[htb]
\centering
\resizebox{!}{!}{\input{eight1DchainsToG2.tikz}}
	\caption{ (Right) We set $G_1$ to be the graph shown in Figure \ref{fig:hamblyKumagaiIntermediateStep}. $G_2$ is constructed as described in Figure \ref{fig:graphDoesNotSatAssump} 
(Right), namely by taking the two copies $G_1 \times \{w_1\}$, $G_1 \times \{w_2\}$ and choosing the subgraph $B_2$ such that $V(B_2)=\{0,2w_1, 2w_2,4 \}$ and $E(B_2)= \emptyset$. Note that $G_2$ is the level-$2$ Hambly-Kumagai diamond graph and denoted by $HK_2$.  (Left)  It is easy to see that $G_{1} \backslash B_2$ is a collection of four $1D$ chains. The two copies of each $1D$ chain in $G_{1} \backslash B_2$ are displayed in Figure.  Each $1D$ chain in $G_{1} \backslash B_2$ is associated with a Jacobi matrix. Due to the mirror symmetry assumption, there are only two different Jacobi matrices. We denote them by $\jacobi_{2,1}$ and  $\jacobi_{2,2}$.}
	\label{fig:eight1DchainsToG2}
\end{figure}
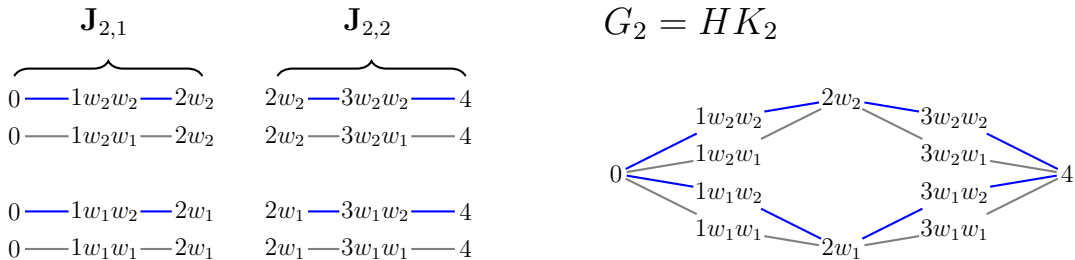

In this section we demonstrate how to apply Theorem \ref{key-thm2} and determine the spectrum of $\hamilton_2$.
To this end, we construct the level-$2$ Hambly-Kumagai diamond graph $HK_2$ using a sequence of discretized projective limit spaces $\{G_0, G_1, G_2\}$ (see Definition \ref{projeLimConDef}) such that $HK_2 = G_2$. We proceed in the following manner:
\begin{enumerate}

    \item[\textbf{(Step 1)}] $G_0$ is initialized to be the induced auxiliary $1D$ chain $D_{4}$ equipped with $\jacobi_2$ given in (\ref{eq: level2Hchain}). Let  $\lambda \in \sigma(\jacobi_2)$ with $v_{\lambda}$ being the corresponding eigenvector. Then  $ P_2^{\ast}v_{\lambda}= \phi_{1}^{\ast}\phi_{2}^{\ast}v_{\lambda}$ gives a corresponding radial $\hamilton_2$-eigenvector on $HK_2$ (by Lemma
\ref{lemmaLifteigenvector} or \ref{genlemmaLiftforInduct}).  Figure \ref{fig:Neu_level2_e1_1st} displays
a radial eigenvector of $\hamilton_2$ corresponding to the eigenvalue $1$. It describes the oscillations of the transversal layers. This step shows $\sigma(\jacobi_2) \subset \sigma(\hamilton_2)$ and generates five radial eigenvectors. Note that $\sigma(\jacobi_2) = \{-2,-1,0,1,2\}$, see Table \ref{tab:level2ChainSpectrumNew} (Left).

    \item[\textbf{(Step 2)}] To construct the graph $G_1$, we proceed as described in Figure
    \ref{fig:hamblyKumagaiIntermediateStep} (Left). This is precisely the situation described in Lemma \ref{lemma: mainLemmaLifttoBranches} (Lifting \& Gluing Lemma), where $G_1$ plays the role of $G_D$ with two branches. Hence, we can lift an eigenvector from $G_0$ to $G_1$ as follows. Recall, when Dirichlet boundary conditions are imposed, we write $\jacobi_2^D$ for the Jacobi matrix. Let $\lambda \in \sigma(\jacobi_2^D)$ with $v^D_{\lambda}$ as the corresponding $\jacobi_2^D$-eigenvector. Then, the vector
    \begin{equation*}
    v_{\lambda} =
  \begin{cases}
    v^D_{\lambda}       & \quad \quad \text{on the $w_1$-branch  } \\
    -v^D_{\lambda}  & \quad \quad \text{on the $w_2$-branch }
  \end{cases}
\end{equation*}
   defines an eigenvector on $G_1$. Lifting $v_{\lambda}$ to $HK_2$ via $ \phi_{2}^{\ast}v_{\lambda}$ gives an eigenvector of $\hamilton_2$ (see Lemma \ref{genlemmaLiftforInduct}).
    Figures \ref{fig:sqrt3eigenvectorsNeuDirNew} (left) \& (middle) display
 eigenvectors of $\hamilton_2$ constructed as described in step 2.
   This step shows $\sigma(\jacobi^D_2) \subset \sigma(\hamilton_2)$ and generates three additional eigenvectors. Note that $\sigma(\jacobi^D_2) = \{-\sqrt{3},0,\sqrt{3}\}$, see Table \ref{tab:level2ChainSpectrumNew} (Middle).

    \item[\textbf{(Step 3)}] To construct the level-$2$ Hambly-Kumagai graph $G_2=HK_2$, we proceed as described in Definition \ref{projeLimConDef} and Figure 	\ref{fig:eight1DchainsToG2}.      {We set $W=\{w_1,w_2\}$ to be the vertical multiplier space.} We choose the subgraph $B_2$ to be edgeless with the vertices set $V(B_2)=\{0,2w_1, 2w_2,4 \}$. In this case, $G_{1} \backslash B_2$ is a collection of four $1D$ chains. The two copies of each $1D$ chain in $G_{1} \backslash B_2$ are displayed in Figure \ref{fig:eight1DchainsToG2} (Left). Gluing the copies at the common boundary vertices gives $G_2=HK_2$, see Figure \ref{fig:eight1DchainsToG2} (Right). Each $1D$ chain in $G_{1} \backslash B_2$ is associated with a Jacobi matrix. Due to the mirror symmetry assumption, it is sufficient to consider one of the four Jacobi matrices. We denote this Jacobi matrix by $\jacobi_{2,1}$, see 
    Figure \ref{fig:eight1DchainsToG2} (Left). It is easy to check that $\jacobi_{2,1}$ is given by
    \begin{equation*}
\jacobi_{2,1} = \left[\begin{matrix}0 & 1 & 0\\ 1 & 0 & \frac{\sqrt{6}}{2}\\0 & \frac{\sqrt{6}}{2} & 1\end{matrix}\right]. 
\end{equation*}
Gluing two copies of a $1D$ chain in $G_{1} \backslash B_2$ generates a situation similar to Lemma \ref{lemma: mainLemmaLifttoBranches} (Lifting \& Gluing Lemma). Hence, we can lift an eigenvector to $G_2$ as follows. Let $\lambda \in \sigma(\jacobi_{2,1}^D)$      { with $v^D_{\lambda}$ being the corresponding $\jacobi_2^D$-eigenvector}. We define
\begin{equation*}
    v_{\lambda} =
  \begin{cases}
    v^D_{\lambda}       & \quad \quad \text{on the $w_1$-branch of a $1D$ chain in $G_{1} \backslash B_2$  } \\
    -v^D_{\lambda}  & \quad \quad \text{on the $w_2$-branch  of the same $1D$ chain in $G_{1} \backslash B_2$} 
    \\
    0       & \quad \quad \text{elsewhere } 
  \end{cases}
\end{equation*}
 It is easy to see that $ v_{\lambda}$ is an eigenvector on $\hamilton_2$. This step shows the inclusion $\sigma(\jacobi_{2,1}^D) \subset \sigma(\hamilton_2)$ and generates four additional eigenvectors, one eigenvector for each $1D$ chain in $G_{1} \backslash B_2$. Note that $\sigma(\jacobi_{2,1}^D) = \{0\}$.
\end{enumerate}
The constructed twelve eigenvectors are orthogonal and therefore $ \sigma (\hamilton_2)= \sigma(\jacobi_2) \cup \sigma(\jacobi_2^D) \cup \sigma(\jacobi_{2,1}^D)$.

\begin{table}[h!]
\centerline{\hbox{ 
\begin{tabular}{lcc}
& Eigenvalue & Multiplicity \\
\hline
1 & -2 & 1 \\
2 & -1 & 1 \\
3 & 0 & 1 \\
4 & 1 & 1 \\
5 & 2 & 1 \\
\hline
\\
\end{tabular}
\quad \quad 
\begin{tabular}{lcc}
& Eigenvalue & Multiplicity \\
\hline
1 & $-\sqrt{3}$ & 2 \\
2 & 0 & 6 \\
3 & $\sqrt{3}$ & 2 \\
\hline
\end{tabular}
\quad \quad 
\begin{tabular}{lcc}
& Eigenvalue & Multiplicity \\
\hline
1 & -2 & 1 \\
2 & $-\sqrt{3}$ & 1 \\
3 & -1 & 1 \\
4 & 0 & 6 \\
5 & 1 & 1 \\
6 & $\sqrt{3}$& 1 \\
7 & 2 & 1 \\ 
\hline
\\
\end{tabular}}}
\caption{Hambly-Kumagai diamond graph of level 2: Eigenvalues table of 
    $\jacobi_2$ (Left), $\jacobi^D_2$ (Middle) and of 
    $\hamilton_2$ (Right).}
\label{tab:level2ChainSpectrumNew}
\end{table}


%
%
%
%
%

\begin{figure}[htbp]

\begin{tabular}{C{.49\textwidth}C{.49\textwidth}}

\subfigure{
  \resizebox{6.cm}{!}{
  \input{symmetry2.tikz}
    }
} &


\subfigure{
  \resizebox{6.5cm}{!}{%
    \input{symmetry3.tikz}
   }
} 

\end{tabular}
\caption{ Eigenvalues of multiplicity five: (Left) With an argument similar to Lemma \ref{lemma: mainLemmaLifttoBranches} (Lifting-\&-Gluing), we can lift a Dirichlet eigenvector $f$ of a particular Jacobi submatrix to the subgraph colored blue. Moreover, we lift the same eigenvector but with the opposite sign to the subgraph colored gray.  The constructed vector is then an eigenvector of $\hamilton_{\ell}$ with the same eigenvalue. In this way, we can construct a total of four eigenvectors to the same eigenvalue. (Right) Lifting the same eigenvector as described in the right-hand side Figure will result in the fifth eigenvector of $\hamilton_{\ell}$ with the same eigenvalue.}
\label{fig:5multiplicityTheorem}
\end{figure}
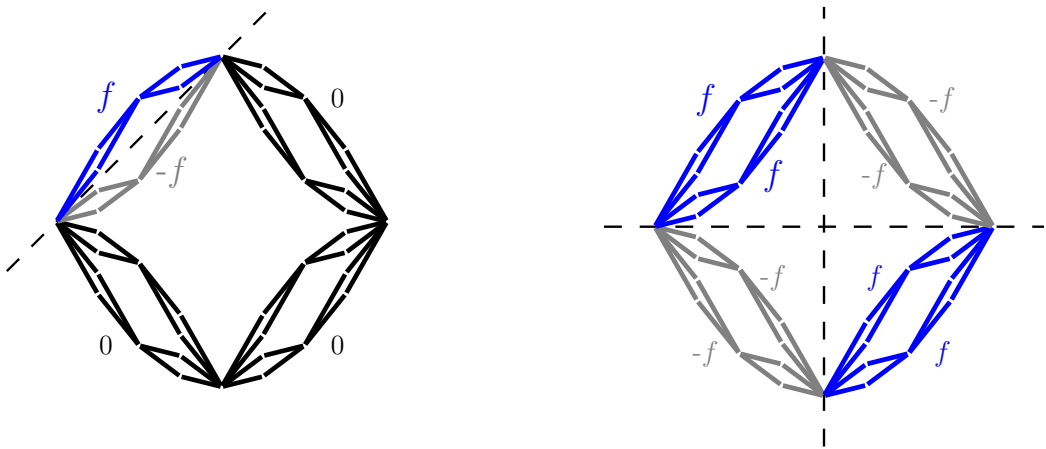

\begin{proposition}
Let $\ell \geq 3$. There exists $\lambda \in \sigma(\hamilton_{\ell})$, such that the multiplicity of $\lambda$ is $5$.
\end{proposition}
\begin{proof}
     {It follows by using an argument similar to Lemma \ref{lemma: mainLemmaLifttoBranches} (Lifting-\&-Gluing). See also Figure \ref{fig:5multiplicityTheorem}
and its caption.}
\end{proof}
We will refer to eigenvectors of $\hamilton_{\ell}$ that are supported on a proper subset of $V(HK_{\ell})$ as localized eigenvectors.
\begin{proposition}
The total number of localized eigenvectors of $\hamilton_{\ell}$ is $\frac{2 \cdot 4^{\ell}+4}{3}-2^{\ell}-2^{\ell-1}$, $\ell \geq 3$. 
\end{proposition}
\begin{proof}
The number of vertices of $HK_{\ell}$ at level $\ell \in \nn$ is $|V(HK_{\ell})| = \frac{2 \cdot 4^{\ell}+4}{3}$.      {Next, the above-described algorithm tells us how to construct the eigenvectors}. In particular, it shows that the only non-localized eigenvectors are the $2^{\ell}+1$ radial eigenvectors and the $2^{\ell-1}-1$ ``fifth''  eigenvectors in Figure \ref{fig:5multiplicityTheorem} (Right).
\end{proof}

For higher levels we can proceed similarly and find the spectrum by considering a collection of Jacobi matrices.      {A convenient representation of the higher levels spectrum is} the
\textit{integrated density of states} of $\hamilton_{\ell}$, that is defined as
\begin{equation*}
N_{\ell}(x) := \frac{\# \{\lambda \leq x \ | \ \lambda \text{ is an eigenvalue of } \hamilton_{\ell} \}}{|V(HK_{\ell})|},
\end{equation*}
where $\#$ counts the number of eigenvalues of $\hamilton_{\ell}$ less or equal than $x$.  Figure \ref{fig:DOSHK1} shows the integrated density of states of $\hamilton_{\ell}$ for both level 6 (Left) and level 7 (Right).
%
%
\begin{figure}[htbp]
\begin{center}
\begin{tabular}{C{.49\textwidth}C{.49\textwidth}} 
\subfigure{
\resizebox{8.9cm}{!}{
\includegraphics{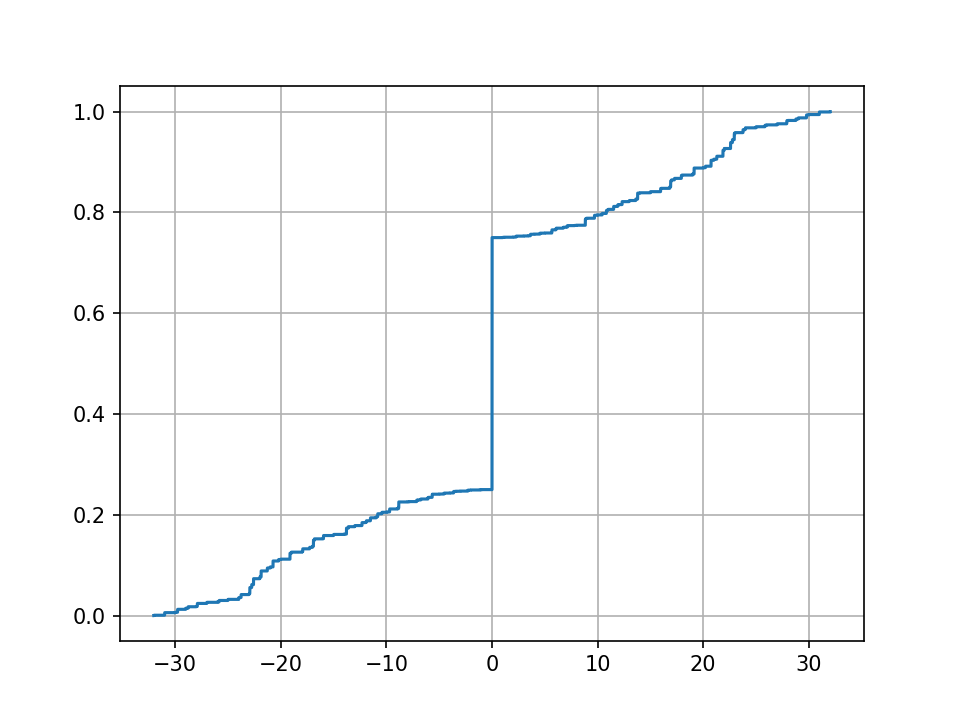}
  }

} &
\subfigure{
\resizebox{8.9cm}{!}{
\includegraphics{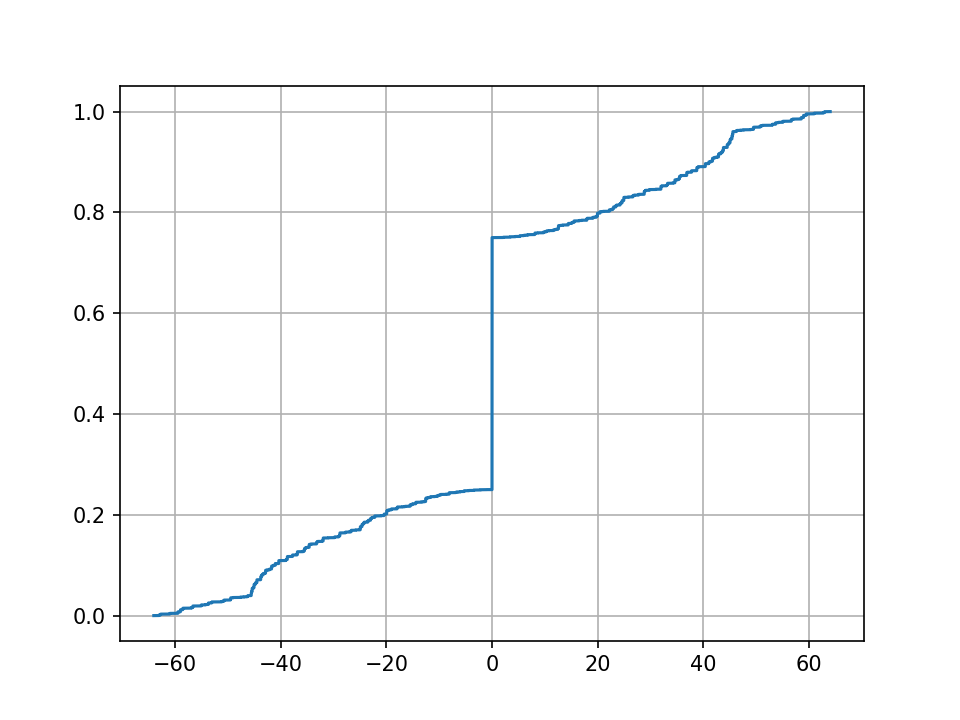}
  }
}
\end{tabular}
\caption{Integrated density of states of $\hamilton_{\ell}$:  Hambly-Kumagai diamond graph of level 6 (Left) and level 7 (Right) .}
\label{fig:DOSHK1}
\end{center}
\end{figure}

\subsection{Lang-Plaut diamond graphs}

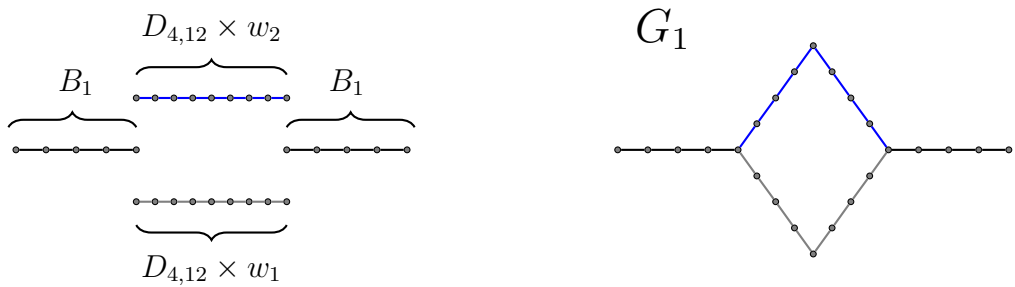
\begin{figure}[htb]
\centering
\resizebox{!}{!}{\input{PlautLangLevel2From0to1Main.tikz}}
	\caption{ (Left) To construct $G_1$ we take two copies of the $1D$ chain $G_0=D_{16}$ (Recall that the vertices set is $V(D_{16})=\{0, \dots , 16\}$) and choose the subgraph $B_1$ such that $G_0 \backslash B_1$ contains only the $1D$ chain with the set of vertices $V(D_{4,12})=\{4,5, \dots, 12\}$. We will refer to this 1D chain as $D_{4,12}$. For the vertical multiplier space we set $W=\{w_1,w_2\}$.  (Right)  $G_1$ is given by gluing the two copies $D_{4,12} \times \{ w_1\}$ and $D_{4,12} \times \{ w_2\}$ together with $B_1$ at their boundary points.}
	\label{fig:plautlangToG1}
\end{figure}

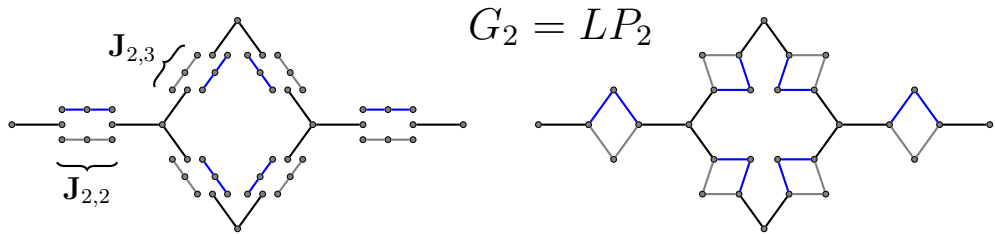
\begin{figure}[htb]
\centering
\resizebox{!}{!}{\input{PlautLangLevel2From1to2Main.tikz}}
	\caption{ (Left) The Jacobi matrices $\jacobi_{2,3}$ and $\jacobi_{2,2}$  are relevant for the computation of  $\sigma (\hamilton_2)$. (Right) Lang-Plaut diamond graphs of level-$2$,  $G_2=HK_2$.}
	\label{fig:plautLangToG2}
\end{figure}

The second model is also an example of a two-point  self-similar graph in the sense of \cite{MT}. 
It is another prominent example of diamond-type graphs, that was investigated in  \cite{LP}. We will refer to this model as \textit{Lang-Plaut diamond graphs}. The following definition gives a formal description of the Lang-Plaut diamond graphs.
\begin{definition}
We refer to a sequence of graphs $\{LP_{\ell}\}_{\ell \geq 0}$ as \textit{Lang-Plaut diamond graphs}, when it is constructed as follows.
\begin{itemize}
    \item $LP_0$ is initialized as the one edge graph connecting a node $x_L$ with another node $x_R$. 
    \item At level $\ell$, we construct $LP_{\ell}$ by segmenting each edge from the previous level $LP_{\ell-1}$ into three new edges. The inner edge of the three new edges is then replaced by two new branches, whereas each new branch is then segmented into two edges.
\end{itemize}
\label{LangPlautDiamond}
\end{definition}
The first four levels of the Lang-Plaut diamond graphs are displayed in \cite[Figure 4, page 10]{2019arXiv190908668D}.
Let $V(LP_{\ell})$ be the vertices set of $LP_{\ell}$. In the same manner as the Hambly-Kumagai diamond graphs, it is easily seen that the transversal decomposition $V(LP_{\ell})=\PI_A^{-1}(0)\cup\PI_A^{-1}(1) \ldots \cup \PI_A^{-1}(N)$ with respect to $A=\{x_L\}$ induces a $1D$ chain $D_{N}$ such that  $N=4^\ell$. The Jacobi matrix associated with $D_{N}$, $N=4^\ell$ is denoted by  $\jacobi_{\ell}$. We lift $\jacobi_{\ell}$ to a Hamiltonian $\hamilton_{\ell}$ on $LP_{\ell}$.


\subsubsection{Spectrum of the Hamiltonian $\hamilton_{2}$}

     {We are going to study this case as we did it in the case of the first model. Namely, we demonstrate how} to apply Theorem \ref{key-thm2} while determining the spectrum of $\hamilton_2$. Similarly, we construct the level-$2$  Lang-Plaut diamond graph $LP_2$ using a sequence of discretized projective limit spaces $\{G_0, G_1, G_2\}$ (see Definition \ref{projeLimConDef}) such that  $LP_2 = G_2$. We proceed as follows:
\begin{enumerate}
   \item[\textbf{(Step 1)}] $G_0$ is initialized to be the induced auxiliary $1D$ chain $D_{16}$ equipped with $\jacobi_2$, whose entries are given in (\ref{KrawtchoukCoupling}). Similar to the first model we can show $\sigma(\jacobi_2) \subset \sigma(\hamilton_2)$ and generate $17$ radial eigenvectors. Note that $\sigma(\jacobi_2) = \{-8,-7,\dots,7,8\}$, see Table \ref{tab:PLlevel2ChainSpectrumNew} (Left).

    \item[\textbf{(step 2)}] To construct the graph $G_1$, we proceed as described in Figure
    \ref{fig:plautlangToG1}. Again, with a similar argument to the Lifting \& Gluing Lemma, we lift an eigenvector from the $1D$ chain $D_{4,12}$ to $G_2$. We denote the Jacobi matrix associated with $D_{4,12}$ by 
    \begin{equation*}
\jacobi_{2,1} =  \left[\begin{matrix}0 & \sqrt{15} & 0 & 0 & 0 & 0 & 0 & 0 & 0\\\sqrt{15} & 0 & \frac{\sqrt{66}}{2} & 0 & 0 & 0 & 0 & 0 & 0\\0 & \frac{\sqrt{66}}{2} & 0 & \frac{\sqrt{70}}{2} & 0 & 0 & 0 & 0 & 0\\0 & 0 & \frac{\sqrt{70}}{2} & 0 & 3 \sqrt{2} & 0 & 0 & 0 & 0\\0 & 0 & 0 & 3 \sqrt{2} & 0 & 3 \sqrt{2} & 0 & 0 & 0\\0 & 0 & 0 & 0 & 3 \sqrt{2} & 0 & \frac{\sqrt{70}}{2} & 0 & 0\\0 & 0 & 0 & 0 & 0 & \frac{\sqrt{70}}{2} & 0 & \frac{\sqrt{66}}{2} & 0\\0 & 0 & 0 & 0 & 0 & 0 & \frac{\sqrt{66}}{2} & 0 & \sqrt{15}\\0 & 0 & 0 & 0 & 0 & 0 & 0 & \sqrt{15} & 0\end{matrix}\right]
\end{equation*}    
   Similar to the first model we can show that $\sigma(\jacobi^D_{2,1}) \subset \sigma(\hamilton_2)$ and generate $7$ additional eigenvectors. The eigenvalues $\jacobi^D_2$ are listed in Table \ref{tab:PLlevel2ChainSpectrumNew} (Middle).
    \item[\textbf{(step 3)}] To construct finally the level-$2$ Lang-Plaut diamond graphs $G_2=HK_2$, we proceed similarly to the first model. The relevant Jacobi matrices $\jacobi_{2,2}$ and $\jacobi_{2,3}$ are indicated in Figure 	\ref{fig:plautLangToG2} (Left).  Again, due to the mirror symmetry, it is sufficient to consider two out of six matrices.  We can  show $\sigma(\jacobi_{2,2}), \ \sigma(\jacobi_{2,3}) \subset \sigma(\hamilton_2)$ and generate 6 additional eigenvectors. Note that $\sigma(\jacobi_{2,2}^D)=\sigma(\jacobi_{2,3}^D) = \{0\}$.
\end{enumerate}
The generated 30 eigenvectors are orthogonal. Hence $\sigma (\hamilton_2)= \sigma(\jacobi_2) \cup \sigma(\jacobi_{2,1}^D) \cup \sigma(\jacobi_{2,2}^D) \cup \sigma(\jacobi_{2,3}^D)$. Figure \ref{fig:DOSHK} shows the integrated density of states of $\hamilton_{\ell}$ for both level 4 (Left) and level 5 (Right).
%

%

\begin{figure}[htbp]
\begin{center}
\begin{tabular}{C{.49\textwidth}C{.49\textwidth}} 
\subfigure{
\resizebox{7.9cm}{!}{
\includegraphics{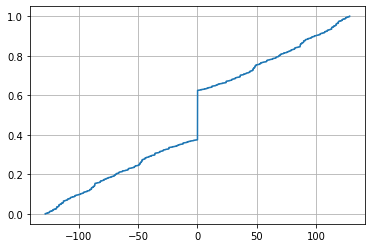}
  }

} &
\subfigure{
\resizebox{7.9cm}{!}{
\includegraphics{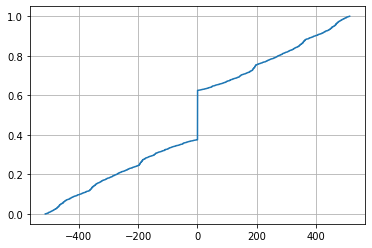}
  }
}
\end{tabular}
\caption{ Integrated density of states of $\hamilton_{\ell}$: Lang-Plaut diamond graph of level 4 (left) and level 5 (right).}
\label{fig:DOSHK}
\end{center}
\end{figure}

\begin{table}[h!]
\centerline{\hbox{ 
\begin{tabular}{lcc}
& Eigenvalue & Multiplicity \\
\hline
1 & -8.0 & 1 \\
2 & -7.0 & 1 \\
3 & -6.0 & 1 \\
4 & -5.0 & 1 \\
5 & -4.0 & 1 \\
6 & -3.0 & 1 \\
7 & -2.0 & 1 \\
8 & -1.0 & 1 \\
9 & 0.0 & 1 \\
10 & 1.0 & 1 \\
11 & 2.0 & 1 \\
12 & 3.0 & 1 \\
13 & 4.0 & 1 \\
14 & 5.0 & 1 \\
15 & 6.0 & 1 \\
16 & 7.0 & 1 \\
17 & 8.0 & 1 \\
\hline
\\
\end{tabular}
\quad \quad 
\begin{tabular}{lcc}
& Eigenvalue & Multiplicity \\
\hline
1 & -7.7536903 & 1 \\
2 & -5.8309519 & 1 \\
3 & -3.1432923 & 1 \\
4 & 0.0 & 1 \\
5 & 3.1432923 & 1 \\
6 & 5.8309519 & 1 \\
7 & 7.7536903 & 1 \\
\hline
\end{tabular}
\quad \quad 
\begin{tabular}{lcc}
& Eigenvalue & Multiplicity \\
\hline
1 & -8.0 & 1 \\
2 & -7.753690 & 1 \\
3 & -7.0 & 1 \\
4 & -6.0 & 1 \\
5 & -5.830951 & 1 \\
6 & -5.0 & 1 \\
7 & -4.0 & 1 \\
8 & -3.143292 & 1 \\
9 & -3.0 & 1 \\
10 & -2.0 & 1 \\
11 & -1.0 & 1 \\
12 & 0.0 & 8 \\
13 & 1.0 & 1 \\
14 & 2.0 & 1 \\
15 & 3.0 & 1 \\
16 & 3.143292 & 1 \\
17 & 4.0 & 1 \\
18 & 5.0 & 1 \\
19 & 5.830951 & 1 \\
20 & 6.0 & 1 \\
21 & 7.0 & 1 \\
22 & 7.753690 & 1 \\
23 & 8.0 & 1 \\
\hline
\\
\end{tabular}}}
\caption{Lang-Plaut diamond graphs  of level 2: Eigenvalues table of $\jacobi_2$ (Left), $\jacobi^D_{2,1}$ (Middle) and of 
    $\hamilton_2$ (Right).}
\label{tab:PLlevel2ChainSpectrumNew}
\end{table}

\section{Further General Geometric Constructions: Two-point self similar graphs}

\label{sec:2PointSG}
In \cite{MT} a broad class of infinite self-similar graphs called \textit{two-point self-similar fractal graphs} was introduced and  
the spectra of the combinatorial- and probabilistic-Laplacians on such graphs were described.
The two-point self-similar fractal graphs are related to the nested fractals with two essential fixed points
\cite{lindstrom1990brownian}. A generalization to self-similar graphs based on a finite symmetric $M$-point model (instead of two points) is constructed in \cite{MT2}.

Following \cite{MT}, we set $M=(V_M,E_M)$ and $G_0=(V_0,E_0)$ to be finite connected graphs, where $M$ is an ordered graph. We fix some $e_0 \in E_M$, which is not a loop, and vertices $\alpha, \beta \in V_M$ and $\alpha_0, \beta_0 \in V_0$, $\alpha \neq \beta$, $\alpha_0 \neq \beta_0$.
\begin{definition}[\cite{MT}, page 393]
A graph $G$ is called \textit{two point self-similar graph} with model graph $M$ and initial graph $G_0$ if the following holds:
\begin{enumerate}
\item There are finite subgraphs $\{ G_n\}_{n \geq 0} $, $G_n = (V_n, E_n)$ such that $G_n \subset G_{n+1}$, $n \geq 0$, and $G = \cup_{n \geq 0} G_n$.
\item For any $n \geq 0$ and $e \in E_M$ there is a graph homomorphism $\Psi_n^{e}:G_n \to G_{n+1}$ such that $G_{n+1}= \cup_{e \in E_M}\Psi_n^{e}(G_n)$ and $\Psi_n^{e_0}$ is the inclusion of $G_n$ to $G_{n+1}$.
\item For all $n \geq 0$ there are two vertices $\alpha_n, \beta_n \in V_n$ such that $\Psi_n^{e}$ restricted to $G_n \backslash \{\alpha_n, \beta_n \}$ is a one-to-one mapping for every $e \in E_M$. Moreover $\Psi_n^{e_1}(V_n \backslash \{\alpha_n, \beta_n \}) \cap \Psi_n^{e_2}(V_n \backslash \{\alpha_n, \beta_n \}) = \emptyset$ if $e_1 \neq e_2$.
\item For $n \geq 1$, there is an injection $\kappa_n:V_M \to V_n$ such that $\alpha_n = \kappa_n(\alpha)$, $\beta_n = \kappa_n(\beta)$ and for every edge $e =(a,b) \in E_M$, $\Psi_{n-1}^{e}(\alpha_{n-1})= \kappa_n(a)$ and $\Psi_{n-1}^{e}(\beta_{n-1})= \kappa_n(b)$.
\end{enumerate}
We say that the vertices $\alpha_n, \beta_n$ are the boundary vertices of $G_n$, i.e. $\partial G_n = \{ \alpha_n, \beta_n\}$ and $int(G_n)=V_n \backslash \{ \alpha_n, \beta_n\}$ are the interior vertices of $G_n$.
\end{definition}

\begin{proposition}
Suppose that the graphs $M$ and $G_0$ satisfy  the assumptions \ref{graphAssumptions}, where the transversal decomposition of $M$ and $G_0$ are with respect $\alpha$ (or $\beta$) and $\alpha_0$ (or $\beta_0$), respectively. Moreover, we assume $\mathbf{deg}_{0}(x)=0$ for all $x \in V_M$. Then the assumptions \ref{graphAssumptions} hold for each $G_i$, $i \geq 0$. And the transversal decomposition of $G_i$ is with respect $\alpha_i$ (or $\beta_i$).
\end{proposition}
\begin{proof}
$G_{\ell+1}$ is obtained by replacing every edge in $M$ by a copy of $G_{\ell}$. Under the assumptions, the transversal decomposition of $M$ with respect $\alpha$ (or $\beta$) is modified by adding the transversal layers of $G_{\ell}$ resulting in a transversal decomposition of $G_{\ell+1}$ with respect $\alpha_{\ell+1}$ (or $\beta_{\ell+1}$).
\end{proof}

\section{Conclusions}

{We study a scheme, proposed in \cite{2019arXiv190908668D}, of lifting $1D$ Hamiltonians realizing a PQST from site $0$ to site $N$ to Hamiltonians on diamond fractal graphs that realize a PQST from a set $A$ of sites to a set $B$ of sites. Note here that under this lifting procedure the set $A$ corresponds to site $0$ and the set $B$ is induced by site $N$. Our construction is a vast theoretical generalization of  some ideas discussed in \cite{P-RKay11}, see also \cite{KLY17}. In addition, we give a constructive algorithm on how to find spectra of such Hamiltonians. We describe the spectrum and eigenfunctions, and give two examples of such Hamiltonians on specific graphs. In particular, we demonstrate  that their spectra contain  multiple eigenvalues with localized eigenfunctions. 
In the future this study will suggest what physical characteristics of the system remains unchanged under the lifting procedure and if there is something that changes. It would also be interesting to know if the lifting procedure can be applied to the case of state transfer for multi-qubits discussed in \cite{Apollaro_2015, YB20, Chetcuti_2020}. }

\section*{Acknowledgments}

This research was supported in part by the University of Connecticut Research Excellence Program, by DOE grant DE-SC$0010339$ and by NSF DMS   grants 1613025 and 2008844.  The authors thank the anonymous referees for insightful remarks and suggestions that helped to improve the presentation of the results and pointed towards new directions.

\bibliographystyle{alpha}
\bibliography{Quantum-State-Refs}

\newpage

\appendix

\section*{Diamond graphs eigenvalue tables}
\enlargethispage*{0.5in}

\begin{center}

\begin{tabular}{lclclclclclcl}
\hline
\\
\multicolumn{12}{c}{Eigenvalues of $\hamilton_3$ (Hambly-Kumagai diamond graph of level $3$)} \\
\hline
\\
 j & Eigenvalue $\lambda_j $ & Multiplicity & j & Eigenvalue $\lambda_j $ & Multiplicity \\
\hline
1 & -4 & 1  							& 9 & 1 & 1 \\
2 & $-\sqrt{\sqrt{46} + 9}$ & 1 		& 10 & $\sqrt{9 - \sqrt{46}}$ & 1 \\
3 & -3 & 1 							& 11 & 2 & 1 \\
4 & $-2\sqrt{2}$ & 5   				& 12 & 2 $\sqrt{2}$ & 5 \\
5 & -2 & 1 							& 13 & 3 & 1 \\
6 & $-\sqrt{9 - \sqrt{46}}$ & 1 		&14 & $\sqrt{\sqrt{46} + 9}$ & 1 \\
7 & -1 & 1 							&15 & 4 & 1 \\
8 & 0 & 22 \\
\hline
\\
\\
\\

\end{tabular}


\begin{tabular*}{\textwidth}{@{\extracolsep{\fill}}lclclclclclclclclclclclcl}
\hline
\\
\multicolumn{12}{c}{Eigenvalues of $\hamilton_4$ (Hambly-Kumagai diamond graph of level $4$)} \\
\hline
\\
  j & $\lambda_j$ & Multipl. & j & $\lambda_j$ & Multipl. & j & $\lambda_j$ & Multipl. & j & $\lambda_j$ & Multipl. \\
  \hline
\\
1 & -8.0 & 1 			       &10 & -4.242640 & 8 			&19 & 1.0 & 1 				&28 & 5.0 & 1 \\
2 & -7.999773 & 1 		&11 & -4.0 & 1 				      &20 & 1.321523 & 1 			&29 & 5.830951 & 8 \\
3 & -7.0 & 1 				&12 & -3.806076 & 1 		       	&21 & 2.0 & 1 				&30 & 5.980884 & 1 \\
4 & -6.996920 & 5 		&13 & -3.0 & 1 			       	&22 & 2.600793 & 5 			&31 & 6.0 & 1 \\
5 & -6.0 & 1 				&14 & -2.600793 & 5 			      &23 & 3.0 & 1 				&32 & 6.996920 & 5 \\
6 & -5.980884 & 1 		&15 & -2.0 & 1 					&24 & 3.806076 & 1 			&33 & 7.0 & 1 \\
7 & -5.830951 & 8 		&16 & -1.321523 & 1 				&25 & 4.0 & 1 				&34 & 7.999773 & 1 \\
8 & -5.0 & 1 				&17 & -1.0 & 1 					&26 & 4.242640 & 8 			&35 & 8.0 & 1 \\
9 & -4.927369 & 5 		&18 & 0.0 & 86 					&27 & 4.927369 & 5 \\
   \hline
   \\
   \\
   \\
\end{tabular*}
 


\newpage 


\enlargethispage*{0.5in}
 

\begin{tabular*}{\textwidth}{@{\extracolsep{\fill}}lclclclclclclclclclclclcl}
\hline
\\
\multicolumn{12}{c}{Eigenvalues of $\hamilton_3$ (Lang-Plaut diamond graph of level $3$)} \\
\hline
\\
  j & $\lambda_j$ & Multipl. & j & $\lambda_j$ & Multipl. & j & $\lambda_j$ & Multipl. & j & $\lambda_j$ & Multipl. \\
\hline
\\
1 & -32.0 & 1                           & 28 & -17.272791 & 1       		& 54 & 0.0 & 44 				& 80 & 17.272791 & 1 \\
2 & -31.999214 & 1                & 29 & -17.0 & 1 				&55 & 1.0 & 1 				&81 & 18.0 & 1 \\
4 & -30.989328 & 1                & 30 & -16.0 & 1 				&56 & 2.0 & 1 				&82 & 19.0 & 1 \\
5 & -30.0 & 1                           & 31 & -15.141636 & 2 			&57 & 3.0 & 1 				&83 & 19.721929 & 1 \\
6 & -29.936434 & 1                & 32 & -15.0 & 1 				&58 & 3.047899 & 1 		&84 & 20.0 & 1 \\
7 & -29.086539 & 4                & 33 & -14.656694 & 1 			&59 & 4.0 & 1 				&85 & 20.48943 & 2 \\
8 & -29.0 & 1                           & 34 & -14.0 & 1 				&60 & 5.0 & 1 				&86 & 21.0 & 1 \\
9 & -28.77233 & 1                  & 35 & -13.0 & 1 				&61 & 6.0 & 1 				&87 & 21.98050 & 1 \\
10 & -28.0 & 1                         & 36 & -12.020182 & 4 			& 62 & 6.066367 & 1 		&88 & 22.0 & 1 \\
11 & -27.422949 & 1              & 37 & -12.0 & 1 				&63 & 7.0 & 1 				&89 & 22.214599 & 4 \\
12 & -27.0 & 1 				 & 38 & -11.89895 & 1 			&64 & 8.0 & 1 				&90 & 23.0 & 1 \\
13 & -26.0 & 1 				 & 39 & -11.0 & 1 				&65 & 8.180098 & 2 		&91 & 24.0 & 1 \\
14 & -25.844467 & 1 	       & 40 & -10.0 & 1 				&66 & 9.0 & 1 				&92 & 24.027192 & 1 \\
15 & -25.0 & 1 				 & 41 & -9.026252 & 1 			&67 & 9.026252 & 1 		&93 & 25.0 & 1 \\
16 & -24.027192 & 1 		& 42 & -9.0 & 1 					&68 & 10.0 & 1 				&94 & 25.844467 & 1 \\
17 & -24.0 & 1 				& 43 & -8.180098 & 2 			&69 & 11.0 & 1 				&95 & 26.0 & 1 \\
18 & -23.0 & 1 				& 44 & -8.0 & 1 					&70 & 11.89895 & 1 		&96 & 27.0 & 1 \\
19 & -22.214599 & 4 		& 45 & -7.0 & 1 					&71 & 12.0 & 1 				&97 & 27.422949 & 1 \\
20 & -22.0 & 1 				&46 & -6.066367 & 1 			&72 & 12.020182 & 4 		&98 & 28.0 & 1 \\
21 & -21.98050 & 1 			&47 & -6.0 & 1 					&73 & 13.0 & 1 				&99 & 28.77233 & 1 \\
22 & -21.0 & 1 				&48 & -5.0 & 1 					&74 & 14.0 & 1 				&100 & 29.0 & 1 \\
23 & -20.48943 & 2 			&49 & -4.0 & 1 					&75 & 14.656694 & 1 		&101 & 29.086539 & 4 \\
24 & -20.0 & 1 				&50 & -3.047899 & 1 			&76 & 15.0 & 1 				&102 & 29.936434 & 1 \\
25 & -19.721929 & 1 		&51 & -3.0 & 1 					&77 & 15.141636 & 2 		&103 & 30.0 & 1 \\
26 & -19.0 & 1 				&52 & -2.0 & 1 					&78 & 16.0 & 1 				&104 & 30.989328 & 1 \\
27 & -18.0 & 1  				&53 & -1.0 & 1 					&79 & 17.0 & 1 				&105 & 31.0 & 1 \\
    &   &&&&&&&&       106 & 31.999214 & 1 \\
    &   &&&&&&&&       107 & 32.0 & 1 \\
   \hline
   \\
   \\
   \\
\end{tabular*}

\end{center}

\end{document}

%% file: ExampleFrom24Dec.tikz.tex
\begin{tikzpicture}

\tikzset{edge_c/.style = {draw=black, very thick}}

\tikzset{vertex_c/.style={circle, draw, fill=blue!50, inner sep=0pt, minimum width=8pt}}

\node[draw=none] at (2.4, -1){\scalebox{1.4}{$x_L$}};
\node[draw=none] at (9.6, -1){\scalebox{1.4}{$x_R$}};

\node[draw=none] at (6, 0.5){\scalebox{1.}{}};


\node[vertex_c] (1w1) at  (4.5,-2) {};
\node[vertex_c] (2w1) at  (6.0,-3) {};
\node[vertex_c] (3w1) at  (7.5,-2) {};


\draw[edge_c] (1w1) to (2w1);
\draw[edge_c] (2w1) to (3w1);

\node[vertex_c] (0w2) at  (3.0,-1) {};
\node[vertex_c] (2w2) at  (6.0,-1) {};
\node[vertex_c] (4w2) at  (9.0,-1) {};


\draw[edge_c] (0w2) to (1w1);
\draw[edge_c] (4w2) to (3w1);


\draw[edge_c] (1w1) to (2w2);
\draw[edge_c] (3w1) to (2w2);

%

\node[vertex_c] (1w3) at  (4.5,0) {};
\node[vertex_c] (2w3) at  (6.0,0) {};
\node[vertex_c] (3w3) at  (7.5,0) {};
\draw[edge_c] (1w3) to (2w3);
\draw[edge_c] (2w3) to (3w3);
%
\draw[edge_c] (0w2) to (1w3);

\draw[edge_c] (4w2) to (3w3);

\end{tikzpicture}

%% file: ExampleFrom24Dec2.tikz.tex
\begin{tikzpicture}

\tikzset{edge_c/.style = {draw=black, very thick}}

\tikzset{vertex_c/.style={circle, draw, fill=blue!50, inner sep=0pt, minimum width=8pt}}

\node[draw=none] at (6, -2.5){\scalebox{1.}{$x_{L,1}$}};
\node[draw=none] at (6, -1.5){\scalebox{1.}{$x_{L,2}$}};
\node[draw=none] at (6, 0.5){\scalebox{1.1}{$x_R$}};


\node[vertex_c] (1w1) at  (4.5,-2) {};
\node[vertex_c] (2w1) at  (6.0,-3) {};
\node[vertex_c] (3w1) at  (7.5,-2) {};


\draw[edge_c] (1w1) to (2w1);
\draw[edge_c] (2w1) to (3w1);

\node[vertex_c] (0w2) at  (3.0,-1) {};
\node[vertex_c] (2w2) at  (6.0,-1) {};
\node[vertex_c] (4w2) at  (9.0,-1) {};


\draw[edge_c] (0w2) to (1w1);
\draw[edge_c] (4w2) to (3w1);


\draw[edge_c] (1w1) to (2w2);
\draw[edge_c] (3w1) to (2w2);

%

\node[vertex_c] (1w3) at  (4.5,0) {};
\node[vertex_c] (2w3) at  (6.0,0) {};
\node[vertex_c] (3w3) at  (7.5,0) {};
\draw[edge_c] (1w3) to (2w3);
\draw[edge_c] (2w3) to (3w3);
%
\draw[edge_c] (0w2) to (1w3);

\draw[edge_c] (4w2) to (3w3);

\end{tikzpicture}

%% file: liftBranchingLemma.tikz.tex
\begin{tikzpicture}


\tikzset{vertex_c/.style={circle,inner sep=0,minimum size=5.0pt}} 
\tikzset{edge_c/.style = {draw=black, thick}}

\node[draw=none] at (-9.0, 1.2){\scalebox{1.5}{$D_4 \times \{ w_3\}$}};

\node[draw=none] at (-9.0, -0.8){\scalebox{1.5}{$D_4 \times \{ w_2\}$}};

\node[draw=none] at (-9.0, -2.8){\scalebox{1.5}{$D_4 \times \{ w_1\}$}};

\node[draw=none] at (2.2, 1.5){\scalebox{2.0}{$G_D$}};


\node[vertex_c] (1w1) at  (4.5,-3) {\scalebox{1.5}{$1w_1$}};
\node[vertex_c] (2w1) at  (6.0,-3) {\scalebox{1.5}{$2w_1$}};
\node[vertex_c] (3w1) at  (7.5,-3) {\scalebox{1.5}{$3w_1$}};


\draw[edge_c] (1w1) to (2w1);
\draw[edge_c] (2w1) to (3w1);

\node[vertex_c] (0w2) at  (3.0,-1) {\scalebox{1.5}{$0$}};
\node[vertex_c] (1w2) at  (4.5,-1) {\scalebox{1.5}{$1w_2$}};
\node[vertex_c] (2w2) at  (6.0,-1) {\scalebox{1.5}{$2w_2$}};
\node[vertex_c] (3w2) at  (7.5,-1) {\scalebox{1.5}{$3w_2$}};
\node[vertex_c] (4w2) at  (9.0,-1) {\scalebox{1.5}{$4$}};
\draw[edge_c] (0w2) to (1w1);
\draw[edge_c] (4w2) to (3w1);

\draw[edge_c] (0w2) to (1w2);
\draw[edge_c] (1w2) to (2w2);
\draw[edge_c] (2w2) to (3w2);
\draw[edge_c] (3w2) to (4w2);
%

\node[vertex_c] (1w3) at  (4.5,1) {\scalebox{1.5}{$1w_3$}};
\node[vertex_c] (2w3) at  (6.0,1)  {\scalebox{1.5}{$2w_3$}};
\node[vertex_c] (3w3) at  (7.5,1)  {\scalebox{1.5}{$3w_3$}};
\draw[edge_c] (1w3) to (2w3);
\draw[edge_c] (2w3) to (3w3);
%
\draw[edge_c] (0w2) to (1w3);

\draw[edge_c] (4w2) to (3w3);

\node[vertex_c] (0w4) at  (-7.0,-3) {\scalebox{1.5}{$0$}};
\node[vertex_c] (1w4) at  (-5.5,-3) {\scalebox{1.5}{$1$}};
\node[vertex_c] (2w4) at  (-4.0,-3) {\scalebox{1.5}{$2$}};
\node[vertex_c] (3w4) at  (-2.5,-3) {\scalebox{1.5}{$3$}};
\node[vertex_c] (4w4) at  (-1.0,-3) {\scalebox{1.5}{$4$}};
\draw[edge_c] (0w4) to (1w4);
\draw[edge_c] (1w4) to (2w4);
\draw[edge_c] (2w4) to (3w4);
\draw[edge_c] (3w4) to (4w4);
%

\node[vertex_c] (0w5) at  (-7.0,-1) {\scalebox{1.5}{$0$}};
\node[vertex_c] (1w5) at  (-5.5,-1) {\scalebox{1.5}{$1$}};
\node[vertex_c] (2w5) at  (-4.0,-1) {\scalebox{1.5}{$2$}};
\node[vertex_c] (3w5) at  (-2.5,-1) {\scalebox{1.5}{$3$}};
\node[vertex_c] (4w5) at  (-1.0,-1) {\scalebox{1.5}{$4$}};
\draw[edge_c] (0w5) to (1w5);
\draw[edge_c] (1w5) to (2w5);
\draw[edge_c] (2w5) to (3w5);
\draw[edge_c] (3w5) to (4w5);
%

\node[vertex_c] (0w6) at  (-7.0,1) {\scalebox{1.5}{$0$}};
\node[vertex_c] (1w6) at  (-5.5,1) {\scalebox{1.5}{$1$}};
\node[vertex_c] (2w6) at  (-4.0,1) {\scalebox{1.5}{$2$}};
\node[vertex_c] (3w6) at  (-2.5,1) {\scalebox{1.5}{$3$}};
\node[vertex_c] (4w6) at  (-1.0,1) {\scalebox{1.5}{$4$}};
\draw[edge_c] (0w6) to (1w6);
\draw[edge_c] (1w6) to (2w6);
\draw[edge_c] (2w6) to (3w6);
\draw[edge_c] (3w6) to (4w6);

\end{tikzpicture}

%% file: two1DchainsToG1.tikz.tex
 \begin{tikzpicture}
\tikzset{vertex/.style={circle,inner sep=0,minimum size=5.0pt}} 
\tikzset{edge/.style = {draw=black, thick}}

\node[draw=none][blue] at (-2.5,-2.5){\scalebox{1}{$G_0 \times \{w_2\}$}};

\node[draw=none][gray] at (-2.5,-3.5){\scalebox{1}{$G_0 \times \{w_1\}$}};

\node[draw=none][blue] at (-2.5,-1.2){\scalebox{1.1}{\scalebox{1.4}{$G_1$}}};


\node[vertex] (0) at  (-1.0,-1.25) {\scalebox{0.8}{0}};
\node[vertex] (1) at  (3,-1.25) {\scalebox{0.8}{$4$}};
\node[vertex] (2) at  (1,-1.75) {\scalebox{0.8}{$2w_1$}};
\node[vertex] (3) at  (1,-0.75) {\scalebox{0.8}{$2w_2$}};
\node[vertex] (5) at  (0,-1.5) {\scalebox{0.8}{$1w_1$}};
\node[vertex] (6) at  (0,-1.0) {\scalebox{0.8}{$1w_2$}};
\node[vertex] (9) at  (2,-1.5) {\scalebox{0.8}{$3w_1$}};
\node[vertex] (10) at  (2,-1.0) {\scalebox{0.8}{$3w_2$}};

\draw[edge][gray] (0) to (5);
\draw[edge][blue] (0) to (6);
\draw[edge][gray] (1) to (9);
\draw[edge][blue] (1) to (10);
\draw[edge][gray] (2) to (5);
\draw[edge][gray] (2) to (9);
\draw[edge][blue] (3) to (6);
\draw[edge][blue] (3) to (10);


\node[vertex] (0d) at  (-1.0,-3.5) {\scalebox{0.8}{0}};
\node[vertex] (0u) at  (-1.0,-2.5) {\scalebox{0.8}{0}};
\node[vertex] (1u) at  (3,-2.5) {\scalebox{0.8}{$4$}};
\node[vertex] (1d) at  (3,-3.5) {\scalebox{0.8}{$4$}};
\node[vertex] (2) at  (1,-3.5) {\scalebox{0.8}{$2$}};
\node[vertex] (3) at  (1,-2.5) {\scalebox{0.8}{$2$}};
\node[vertex] (5) at  (0,-3.5) {\scalebox{0.8}{$1$}};
\node[vertex] (6) at  (0,-2.5) {\scalebox{0.8}{$1$}};
\node[vertex] (9) at  (2,-3.5) {\scalebox{0.8}{$3$}};
\node[vertex] (10) at  (2,-2.5) {\scalebox{0.8}{$3$}};

\draw[edge][gray] (0d) to (5);
\draw[edge][blue] (0u) to (6);
\draw[edge][gray] (1d) to (9);
\draw[edge][blue] (1u) to (10);
\draw[edge][gray] (2) to (5);
\draw[edge][gray] (2) to (9);
\draw[edge][blue] (3) to (6);
\draw[edge][blue] (3) to (10);

\end{tikzpicture}

  

%% file: mappingDiagFromBen.tikz.tex
 \xymatrix{ 
		&V(G_1)\times W \ar[dr]_{} \ar[r]^{ \ \ \pi_{2}} & V(G_2) \ar[d]^{\phi_{2}}&&\\
		&V(G_0)\times W \ar[dr]_{} \ar[r]^{ \ \pi_{1}} & V(G_1)\ar[d]^{\phi_{1}}&&\\ 
		& & V(G_0)&&\\
	}

%% file: graphDoesNotSatAssump.tikz.tex
\begin{tikzpicture}
\tikzset{vertex/.style={circle,inner sep=0,minimum size=5.0pt}} 
\tikzset{edge/.style = {draw=black, thick}}

\node[draw=none] at (5,-0.9){\scalebox{1}{\scalebox{1.5}{$G_2$}}};

\node[draw=none] at (-3,-0.9){\scalebox{1}{\scalebox{1.5}{$\tilde{G}_2$}}};


\node[vertex] (0) at  (-3.0,-2) {\scalebox{0.8}{0}};
\node[vertex] (1) at  (3,-1.5) {\scalebox{0.8}{$4w_2$}};
\node[vertex] (2) at  (0,-2) {\scalebox{0.8}{$2w_1$}};
\node[vertex] (3) at  (0,-1) {\scalebox{0.8}{$2w_2w_2$}};
\node[vertex] (5) at  (-1.5,-1.75) {\scalebox{0.8}{$1w_1w_2$}};
\node[vertex] (6) at  (-1.5,-1.25) {\scalebox{0.8}{$1w_2w_2$}};
\node[vertex] (9) at  (1.5,-1.75) {\scalebox{0.8}{$3w_1w_2$}};
\node[vertex] (10) at  (1.5,-1.25) {\scalebox{0.8}{$3w_2w_2$}};

\draw[edge][blue] (0) to (5);
\draw[edge][blue] (0) to (6);
\draw[edge][blue] (1) to (9);
\draw[edge][blue] (1) to (10);
\draw[edge][blue] (2) to (5);
\draw[edge][blue] (2) to (9);
\draw[edge][blue] (3) to (6);
\draw[edge][blue] (3) to (10);


\node[vertex] (0) at  (-3.0,-2.0) {\scalebox{0.8}{0}};
\node[vertex] (1) at  (3,-2.5) {\scalebox{0.8}{$4w_1$}};
\node[vertex] (2) at  (0,-3.0) {\scalebox{0.8}{$2w_2w_1$}};
\node[vertex] (3) at  (0,-2.0) {\scalebox{0.8}{$2w_1$}};
\node[vertex] (5) at  (-1.5,-2.75) {\scalebox{0.8}{$1w_2w_1$}};
\node[vertex] (6) at  (-1.5,-2.25) {\scalebox{0.8}{$1w_1w_1$}};
\node[vertex] (9) at  (1.5,-2.75) {\scalebox{0.8}{$3w_2w_1$}};
\node[vertex] (10) at  (1.5,-2.25) {\scalebox{0.8}{$3w_1w_1$}};

\draw[edge][gray] (0) to (5);
\draw[edge][gray] (0) to (6);
\draw[edge][gray] (1) to (9);
\draw[edge][gray] (1) to (10);
\draw[edge][gray] (2) to (5);
\draw[edge][gray] (2) to (9);
\draw[edge][gray] (3) to (6);
\draw[edge][gray] (3) to (10);

\node[vertex] (0) at  (5,-2.0) {\scalebox{0.8}{$0$}};
\node[vertex] (1) at  (11,-2.0) {\scalebox{0.8}{$4$}};
\node[vertex] (2) at  (8,-3.0) {\scalebox{0.8}{$2w_1$}};
\node[vertex] (3) at  (8,-1.0) {\scalebox{0.8}{$2w_2$}};
\node[vertex] (4) at  (6.5,-2.75) {\scalebox{0.8}{$1w_1w_1$}};
\node[vertex] (5) at  (6.5,-2.25) {\scalebox{0.8}{$1w_1w_2$}};
\node[vertex] (6) at  (6.5,-1.75) {\scalebox{0.8}{$1w_2w_1$}};
\node[vertex] (7) at  (6.5,-1.25) {\scalebox{0.8}{$1w_2w_2$}};
\node[vertex] (8) at  (9.5,-2.75) {\scalebox{0.8}{$3w_1w_1$}};
\node[vertex] (9) at  (9.5,-2.25) {\scalebox{0.8}{$3w_1w_2$}};
\node[vertex] (10) at  (9.5,-1.75) {\scalebox{0.8}{$3w_2w_1$}};
\node[vertex] (11) at  (9.5,-1.25) {\scalebox{0.8}{$3w_2w_2$}};
\draw[edge][gray] (0) to (4);
\draw[edge][blue] (0) to (5);
\draw[edge][gray] (0) to (6);
\draw[edge][blue] (0) to (7);
\draw[edge][gray] (1) to (8);
\draw[edge][blue] (1) to (9);
\draw[edge][gray] (1) to (10);
\draw[edge][blue] (1) to (11);
\draw[edge][gray] (2) to (4);
\draw[edge][blue] (2) to (5);
\draw[edge][gray] (2) to (8);
\draw[edge][blue] (2) to (9);
\draw[edge][gray] (3) to (6);
\draw[edge][blue] (3) to (7);
\draw[edge][gray] (3) to (10);
\draw[edge][blue] (3) to (11);

\end{tikzpicture}

%% file: Neu_level2_eSqrt3_1st.tikz.tex
\begin{tikzpicture}
\tikzset{vertex/.style={circle, inner sep=0,minimum size=1.0pt}} 
\tikzset{edge/.style = {draw=black, thick}}
\node[vertex] (0) at  (-1.0,0.0) {\scalebox{.5}{$0$}};
\node[vertex] (1) at  (1.0,0.0) {\scalebox{.5}{$0$}};
\node[vertex] (2) at  (0.0,-1.0) {\scalebox{.5}{$-\sqrt{2}$}};
\node[vertex] (3) at  (0.0,1.0) {\scalebox{.5}{$\sqrt{2}$}};
\node[vertex] (4) at  (-0.5,-0.75) {\scalebox{.5}{$-1$}};
\node[vertex] (5) at  (-0.5,-0.25) {\scalebox{.5}{$-1$}};
\node[vertex] (6) at  (-0.5,0.25) {\scalebox{.5}{$1$}};
\node[vertex] (7) at  (-0.5,0.75) {\scalebox{.5}{$1$}};
\node[vertex] (8) at  (0.5,-0.75) {\scalebox{.5}{$-1$}};
\node[vertex] (9) at  (0.5,-0.25) {\scalebox{.5}{$-1$}};
\node[vertex] (10) at  (0.5,0.25) {\scalebox{.5}{$1$}};
\node[vertex] (11) at  (0.5,0.75) {\scalebox{.5}{$1$}};
\draw[edge] (0) to (4);
\draw[edge] (0) to (5);
\draw[edge] (0) to (6);
\draw[edge] (0) to (7);
\draw[edge] (1) to (8);
\draw[edge] (1) to (9);
\draw[edge] (1) to (10);
\draw[edge] (1) to (11);
\draw[edge] (2) to (4);
\draw[edge] (2) to (5);
\draw[edge] (2) to (8);
\draw[edge] (2) to (9);
\draw[edge] (3) to (6);
\draw[edge] (3) to (7);
\draw[edge] (3) to (10);
\draw[edge] (3) to (11);
\end{tikzpicture}

%% file: Neu_level2_e0_6th.tikz.tex
\begin{tikzpicture}
\tikzset{vertex/.style={circle, inner sep=0,minimum size=1.0pt}} 
\tikzset{edge/.style = {draw=black, thick}}
\node[vertex] (0) at  (-1.0,0.0) {\scalebox{.5}{$0$}};
\node[vertex] (1) at  (1.0,0.0) {\scalebox{.5}{$0$}};
\node[vertex] (2) at  (0.0,-1.0) {\scalebox{.5}{$0$}};
\node[vertex] (3) at  (0.0,1.0) {\scalebox{.5}{$0$}};
\node[vertex] (4) at  (-0.5,-0.75) {\scalebox{.5}{$-1$}};
\node[vertex] (5) at  (-0.5,-0.25) {\scalebox{.5}{$-1$}};
\node[vertex] (6) at  (-0.5,0.25) {\scalebox{.5}{$1$}};
\node[vertex] (7) at  (-0.5,0.75) {\scalebox{.5}{$1$}};
\node[vertex] (8) at  (0.5,-0.75) {\scalebox{.5}{$1$}};
\node[vertex] (9) at  (0.5,-0.25) {\scalebox{.5}{$1$}};
\node[vertex] (10) at  (0.5,0.25) {\scalebox{.5}{$-1$}};
\node[vertex] (11) at  (0.5,0.75) {\scalebox{.5}{$-1$}};

\draw[edge] (0) to (4);
\draw[edge] (0) to (5);
\draw[edge] (0) to (6);
\draw[edge] (0) to (7);
\draw[edge] (1) to (8);
\draw[edge] (1) to (9);
\draw[edge] (1) to (10);
\draw[edge] (1) to (11);
\draw[edge] (2) to (4);
\draw[edge] (2) to (5);
\draw[edge] (2) to (8);
\draw[edge] (2) to (9);
\draw[edge] (3) to (6);
\draw[edge] (3) to (7);
\draw[edge] (3) to (10);
\draw[edge] (3) to (11);
\end{tikzpicture}

%% file: Neu_level2_e0_3rd.tikz.tex
\begin{tikzpicture}
\tikzset{vertex/.style={circle, inner sep=0,minimum size=1.0pt}} 
\tikzset{edge/.style = {draw=black, thick}}
\node[vertex] (0) at  (-1.0,0.0) {\scalebox{.5}{$0$}};
\node[vertex] (1) at  (1.0,0.0) {\scalebox{.5}{$0$}};
\node[vertex] (2) at  (0.0,-1.0) {\scalebox{.5}{$0$}};
\node[vertex] (3) at  (0.0,1.0) {\scalebox{.5}{$0$}};
\node[vertex] (4) at  (-0.5,-0.75) {\scalebox{.5}{$0$}};
\node[vertex] (5) at  (-0.5,-0.25) {\scalebox{.5}{$0$}};
\node[vertex] (6) at  (-0.5,0.25) {\scalebox{.5}{$-1$}};
\node[vertex] (7) at  (-0.5,0.75) {\scalebox{.5}{$1$}};
\node[vertex] (8) at  (0.5,-0.75) {\scalebox{.5}{$0$}};
\node[vertex] (9) at  (0.5,-0.25) {\scalebox{.5}{$0$}};
\node[vertex] (10) at  (0.5,0.25) {\scalebox{.5}{$0$}};
\node[vertex] (11) at  (0.5,0.75) {\scalebox{.5}{$0$}};
\draw[edge] (0) to (4);
\draw[edge] (0) to (5);
\draw[edge] (0) to (6);
\draw[edge] (0) to (7);
\draw[edge] (1) to (8);
\draw[edge] (1) to (9);
\draw[edge] (1) to (10);
\draw[edge] (1) to (11);
\draw[edge] (2) to (4);
\draw[edge] (2) to (5);
\draw[edge] (2) to (8);
\draw[edge] (2) to (9);
\draw[edge] (3) to (6);
\draw[edge] (3) to (7);
\draw[edge] (3) to (10);
\draw[edge] (3) to (11);
\end{tikzpicture}

%% file: Neu_level2_e1_1st.tikz.tex
\begin{tikzpicture}
\tikzset{vertex/.style={circle, inner sep=0,minimum size=1.0pt}} 
\tikzset{edge/.style = {draw=black, thick}}
\node[vertex] (0) at  (-1.0,0.0) {\scalebox{.5}{$-1$}};
\node[vertex] (1) at  (1.0,0.0) {\scalebox{.5}{$1$}};
\node[vertex] (2) at  (0.0,-1.0) {\scalebox{.5}{$0$}};
\node[vertex] (3) at  (0.0,1.0) {\scalebox{.5}{$0$}};
\node[vertex] (4) at  (-0.5,-0.75) {\scalebox{.5}{$-1$}};
\node[vertex] (5) at  (-0.5,-0.25) {\scalebox{.5}{$-1$}};
\node[vertex] (6) at  (-0.5,0.25) {\scalebox{.5}{$-1$}};
\node[vertex] (7) at  (-0.5,0.75) {\scalebox{.5}{$-1$}};
\node[vertex] (8) at  (0.5,-0.75) {\scalebox{.5}{$1$}};
\node[vertex] (9) at  (0.5,-0.25) {\scalebox{.5}{$1$}};
\node[vertex] (10) at  (0.5,0.25) {\scalebox{.5}{$1$}};
\node[vertex] (11) at  (0.5,0.75) {\scalebox{.5}{$1$}};
\draw[edge] (0) to (4);
\draw[edge] (0) to (5);
\draw[edge] (0) to (6);
\draw[edge] (0) to (7);
\draw[edge] (1) to (8);
\draw[edge] (1) to (9);
\draw[edge] (1) to (10);
\draw[edge] (1) to (11);
\draw[edge] (2) to (4);
\draw[edge] (2) to (5);
\draw[edge] (2) to (8);
\draw[edge] (2) to (9);
\draw[edge] (3) to (6);
\draw[edge] (3) to (7);
\draw[edge] (3) to (10);
\draw[edge] (3) to (11);
\end{tikzpicture}

%% file: eight1DchainsToG2.tikz.tex

\begin{tikzpicture}
\tikzset{vertex/.style={circle,inner sep=0,minimum size=5.0pt}} 
\tikzset{edge/.style = {draw=black, thick}}

\node[draw=none] at (6,0){\scalebox{1}{\scalebox{1.2}{$G_2=HK_2$}}};

\node[draw=none] at (-1.8,0){\scalebox{1}{$\jacobi_{2,1}$}};

\node[draw=none] at (1.7,0){\scalebox{1}{$\jacobi_{2,2}$}};
\usetikzlibrary{decorations.pathreplacing}

\draw[thick] [decorate,decoration= {brace,amplitude=0.2cm} ] (-3,-0.7) - - (-0.5,-0.7);

\draw[thick] [decorate,decoration= {brace,amplitude=0.2cm} ] (0.5,-0.7) - - (3,-0.7);



\node[vertex] (0a) at  (-3.0,-1) {\scalebox{0.8}{0}};

\node[vertex] (0b) at  (-3.0,-1.5) {\scalebox{0.8}{0}};

\node[vertex] (0c) at  (-3.0,-2.5) {\scalebox{0.8}{0}};

\node[vertex] (0d) at  (-3.0,-3) {\scalebox{0.8}{0}};


\node[vertex] (1a) at  (3,-1) {\scalebox{0.8}{$4$}};

\node[vertex] (1b) at  (3,-1.5) {\scalebox{0.8}{$4$}};


\node[vertex] (2b) at  (-0.6,-1.5) {\scalebox{0.8}{$2w_2$}};

\node[vertex] (2c) at  (-0.6,-2.5) {\scalebox{0.8}{$2w_1$}};

\node[vertex] (2d) at  (-0.6,-3) {\scalebox{0.8}{$2w_1$}};


\node[vertex] (3n) at  (-0.6,-1) {\scalebox{0.8}{$2w_2$}};

\node[vertex] (3pa) at  (0.6,-1) {\scalebox{0.8}{$2w_2$}};

\node[vertex] (3pb) at  (0.6,-1.5) {\scalebox{0.8}{$2w_2$}};

\node[vertex] (3pc) at  (0.6,-2.5) {\scalebox{0.8}{$2w_1$}};

\node[vertex] (3pd) at  (0.6,-3) {\scalebox{0.8}{$2w_1$}};


\node[vertex] (6a) at  (-1.8,-1)
{\scalebox{0.8}{$1w_2w_2$}};

\node[vertex] (6b) at  (-1.8,-1.5)
{\scalebox{0.8}{$1w_2w_1$}};

\node[vertex] (1c) at  (3,-2.5) {\scalebox{0.8}{$4$}};

\node[vertex] (1d) at  (3,-3) {\scalebox{0.8}{$4$}};


\node[vertex] (10a) at  (1.8,-1) {\scalebox{0.8}{$3w_2w_2$}};

\node[vertex] (10b) at  (1.8,-1.5) {\scalebox{0.8}{$3w_2w_1$}};

\node[vertex] (10c) at  (1.8,-2.5) {\scalebox{0.8}{$3w_1w_2$}};

\node[vertex] (10d) at  (1.8,-3) {\scalebox{0.8}{$3w_1w_1$}};

\draw[edge][blue] (1a) to (10a);
\draw[edge][gray] (1b) to (10b);

\draw[edge][blue] (1c) to (10c);
\draw[edge][gray] (1d) to (10d);

\draw[edge][blue] (3pc) to (10c);
\draw[edge][gray] (3pd) to (10d);

\draw[edge][blue] (3pa) to (10a);
\draw[edge][gray] (3pb) to (10b);

\draw[edge][blue] (0a) to (6a);
\draw[edge][gray] (0b) to (6b);

\draw[edge][blue] (3n) to (6a);
\draw[edge][gray] (2b) to (6b);




\node[vertex] (6c) at  (-1.8,-2.5) {\scalebox{0.8}{$1w_1w_2$}};

\node[vertex] (6d) at  (-1.8,-3) {\scalebox{0.8}{$1w_1w_1$}};



\draw[edge][blue] (0c) to (6c);
\draw[edge][gray] (0d) to (6d);
\draw[edge][blue] (2c) to (6c);
\draw[edge][gray] (2d) to (6d);

\node[vertex] (0) at  (5,-2.0) {\scalebox{0.8}{$0$}};
\node[vertex] (1) at  (11,-2.0) {\scalebox{0.8}{$4$}};
\node[vertex] (2) at  (8,-3.0) {\scalebox{0.8}{$2w_1$}};
\node[vertex] (3) at  (8,-1.0) {\scalebox{0.8}{$2w_2$}};
\node[vertex] (4) at  (6.5,-2.75) {\scalebox{0.8}{$1w_1w_1$}};
\node[vertex] (5) at  (6.5,-2.25) {\scalebox{0.8}{$1w_1w_2$}};
\node[vertex] (6) at  (6.5,-1.75) {\scalebox{0.8}{$1w_2w_1$}};
\node[vertex] (7) at  (6.5,-1.25) {\scalebox{0.8}{$1w_2w_2$}};
\node[vertex] (8) at  (9.5,-2.75) {\scalebox{0.8}{$3w_1w_1$}};
\node[vertex] (9) at  (9.5,-2.25) {\scalebox{0.8}{$3w_1w_2$}};
\node[vertex] (10) at  (9.5,-1.75) {\scalebox{0.8}{$3w_2w_1$}};
\node[vertex] (11) at  (9.5,-1.25) {\scalebox{0.8}{$3w_2w_2$}};
\draw[edge][gray] (0) to (4);
\draw[edge][blue] (0) to (5);
\draw[edge][gray] (0) to (6);
\draw[edge][blue] (0) to (7);
\draw[edge][gray] (1) to (8);
\draw[edge][blue] (1) to (9);
\draw[edge][gray] (1) to (10);
\draw[edge][blue] (1) to (11);
\draw[edge][gray] (2) to (4);
\draw[edge][blue] (2) to (5);
\draw[edge][gray] (2) to (8);
\draw[edge][blue] (2) to (9);
\draw[edge][gray] (3) to (6);
\draw[edge][blue] (3) to (7);
\draw[edge][gray] (3) to (10);
\draw[edge][blue] (3) to (11);

\end{tikzpicture}

%% file: symmetry2.tikz.tex
\begin{tikzpicture}
\tikzset{vertex/.style={inner sep=0,minimum size=0.005pt}} 
\tikzset{edge/.style = {draw=black, thick}}
\node[vertex] (0) at  (-1.0,0.0) {};
\node[vertex] (1) at  (1.0,0.0) {};
\node[vertex] (2) at  (0.0,-1.0) {};
\node[vertex] (3) at  (0.0,1.0) {};
\node[vertex] (4) at  (-0.5,-0.75) {};
\node[vertex] (5) at  (-0.5,-0.25) {};
\node[vertex] (6) at  (-0.5,0.25) {};
\node[vertex] (7) at  (-0.5,0.75) {};
\node[vertex] (8) at  (0.5,-0.75) {};
\node[vertex] (9) at  (0.5,-0.25) {};
\node[vertex] (10) at  (0.5,0.25) {};
\node[vertex] (11) at  (0.5,0.75) {};
\node[vertex] (12) at  (-0.75,-0.4375) {};
\node[vertex] (13) at  (-0.75,-0.3125) {};
\node[vertex] (14) at  (-0.75,-0.1875) {};
\node[vertex] (15) at  (-0.75,-0.0625) {};
\node[vertex] (16) at  (-0.75,0.0625) {};
\node[vertex] (17) at  (-0.75,0.1875) {};
\node[vertex] (18) at  (-0.75,0.3125) {};
\node[vertex] (19) at  (-0.75,0.4375) {};
\node[vertex] (20) at  (0.75,-0.4375) {};
\node[vertex] (21) at  (0.75,-0.3125) {};
\node[vertex] (22) at  (0.75,-0.1875) {};
\node[vertex] (23) at  (0.75,-0.0625) {};
\node[vertex] (24) at  (0.75,0.0625) {};
\node[vertex] (25) at  (0.75,0.1875) {};
\node[vertex] (26) at  (0.75,0.3125) {};
\node[vertex] (27) at  (0.75,0.4375) {};
\node[vertex] (28) at  (-0.25,-0.9375) {};
\node[vertex] (29) at  (-0.25,-0.8125) {};
\node[vertex] (30) at  (-0.25,-0.6875) {};
\node[vertex] (31) at  (-0.25,-0.5625) {};
\node[vertex] (32) at  (0.25,-0.9375) {};
\node[vertex] (33) at  (0.25,-0.8125) {};
\node[vertex] (34) at  (0.25,-0.6875) {};
\node[vertex] (35) at  (0.25,-0.5625) {};
\node[vertex] (36) at  (-0.25,0.5625) {};
\node[vertex] (37) at  (-0.25,0.6875) {};
\node[vertex] (38) at  (-0.25,0.8125) {};
\node[vertex] (39) at  (-0.25,0.9375) {};
\node[vertex] (40) at  (0.25,0.5625) {};
\node[vertex] (41) at  (0.25,0.6875) {};
\node[vertex] (42) at  (0.25,0.8125) {};
\node[vertex] (43) at  (0.25,0.9375) {};

\draw[dashed] (-1.3,-0.3) -- (0.3,1.3);

\node[draw=none][blue] at (-0.7, 0.75){\scalebox{.5}{$f$}};

\node[draw=none][gray] at (-0.3, 0.3){\scalebox{.5}{-$f$}};

\node[draw=none] at (0.7, 0.75){\scalebox{.4}{$0$}};
\node[draw=none] at (0.7, -0.75){\scalebox{.4}{$0$}};
\node[draw=none] at (-0.7, -0.75){\scalebox{.4}{$0$}};

\draw[edge] (0) to (12);
\draw[edge] (0) to (13);
\draw[edge] (0) to (14);
\draw[edge] (0) to (15);
\draw[edge][gray] (0) to (16);
\draw[edge][gray] (0) to (17);
\draw[edge][blue] (0) to (18);
\draw[edge][blue] (0) to (19);
\draw[edge] (1) to (20);
\draw[edge] (1) to (21);
\draw[edge] (1) to (22);
\draw[edge] (1) to (23);
\draw[edge] (1) to (24);
\draw[edge] (1) to (25);
\draw[edge] (1) to (26);
\draw[edge] (1) to (27);
\draw[edge] (2) to (28);
\draw[edge] (2) to (29);
\draw[edge] (2) to (30);
\draw[edge] (2) to (31);
\draw[edge] (2) to (32);
\draw[edge] (2) to (33);
\draw[edge] (2) to (34);
\draw[edge] (2) to (35);
\draw[edge][gray] (3) to (36);
\draw[edge][gray] (3) to (37);
\draw[edge][blue] (3) to (38);
\draw[edge][blue] (3) to (39);
\draw[edge] (3) to (40);
\draw[edge] (3) to (41);
\draw[edge] (3) to (42);
\draw[edge] (3) to (43);
\draw[edge] (4) to (12);
\draw[edge] (4) to (13);
\draw[edge] (4) to (28);
\draw[edge] (4) to (29);
\draw[edge] (5) to (14);
\draw[edge] (5) to (15);
\draw[edge] (5) to (30);
\draw[edge] (5) to (31);
\draw[edge][gray]  (6) to (16);
\draw[edge][gray]  (6) to (17);
\draw[edge][gray]  (6) to (36);
\draw[edge][gray] (6) to (37);
\draw[edge][blue] (7) to (18);
\draw[edge][blue] (7) to (19);
\draw[edge][blue] (7) to (38);
\draw[edge][blue] (7) to (39);
\draw[edge] (8) to (20);
\draw[edge] (8) to (21);
\draw[edge] (8) to (32);
\draw[edge] (8) to (33);
\draw[edge] (9) to (22);
\draw[edge] (9) to (23);
\draw[edge] (9) to (34);
\draw[edge] (9) to (35);
\draw[edge] (10) to (24);
\draw[edge] (10) to (25);
\draw[edge] (10) to (40);
\draw[edge] (10) to (41);
\draw[edge] (11) to (26);
\draw[edge] (11) to (27);
\draw[edge] (11) to (42);
\draw[edge] (11) to (43);
\end{tikzpicture}

%% file: symmetry3.tikz.tex
\begin{tikzpicture}
\tikzset{vertex/.style={inner sep=0,minimum size=0.005pt}} 
\tikzset{edge/.style = {draw=black, thick}}
\node[vertex] (0) at  (-1.0,0.0) {};
\node[vertex] (1) at  (1.0,0.0) {};
\node[vertex] (2) at  (0.0,-1.0) {};
\node[vertex] (3) at  (0.0,1.0) {};
\node[vertex] (4) at  (-0.5,-0.75) {};
\node[vertex] (5) at  (-0.5,-0.25) {};
\node[vertex] (6) at  (-0.5,0.25) {};
\node[vertex] (7) at  (-0.5,0.75) {};
\node[vertex] (8) at  (0.5,-0.75) {};
\node[vertex] (9) at  (0.5,-0.25) {};
\node[vertex] (10) at  (0.5,0.25) {};
\node[vertex] (11) at  (0.5,0.75) {};
\node[vertex] (12) at  (-0.75,-0.4375) {};
\node[vertex] (13) at  (-0.75,-0.3125) {};
\node[vertex] (14) at  (-0.75,-0.1875) {};
\node[vertex] (15) at  (-0.75,-0.0625) {};
\node[vertex] (16) at  (-0.75,0.0625) {};
\node[vertex] (17) at  (-0.75,0.1875) {};
\node[vertex] (18) at  (-0.75,0.3125) {};
\node[vertex] (19) at  (-0.75,0.4375) {};
\node[vertex] (20) at  (0.75,-0.4375) {};
\node[vertex] (21) at  (0.75,-0.3125) {};
\node[vertex] (22) at  (0.75,-0.1875) {};
\node[vertex] (23) at  (0.75,-0.0625) {};
\node[vertex] (24) at  (0.75,0.0625) {};
\node[vertex] (25) at  (0.75,0.1875) {};
\node[vertex] (26) at  (0.75,0.3125) {};
\node[vertex] (27) at  (0.75,0.4375) {};
\node[vertex] (28) at  (-0.25,-0.9375) {};
\node[vertex] (29) at  (-0.25,-0.8125) {};
\node[vertex] (30) at  (-0.25,-0.6875) {};
\node[vertex] (31) at  (-0.25,-0.5625) {};
\node[vertex] (32) at  (0.25,-0.9375) {};
\node[vertex] (33) at  (0.25,-0.8125) {};
\node[vertex] (34) at  (0.25,-0.6875) {};
\node[vertex] (35) at  (0.25,-0.5625) {};
\node[vertex] (36) at  (-0.25,0.5625) {};
\node[vertex] (37) at  (-0.25,0.6875) {};
\node[vertex] (38) at  (-0.25,0.8125) {};
\node[vertex] (39) at  (-0.25,0.9375) {};
\node[vertex] (40) at  (0.25,0.5625) {};
\node[vertex] (41) at  (0.25,0.6875) {};
\node[vertex] (42) at  (0.25,0.8125) {};
\node[vertex] (43) at  (0.25,0.9375) {};

\draw[dashed] (-1.3,0) -- (1.3,0);

\draw[dashed] (0,-1.3) -- (0,1.3);

\node[draw=none][blue] at (-0.7, 0.75){\scalebox{.5}{$f$}};
\node[draw=none][blue] at (-0.3, 0.3){\scalebox{.5}{$f$}};

\node[draw=none][gray] at (0.7, 0.75){\scalebox{.4}{-$f$}};
\node[draw=none][gray] at (0.3, 0.3){\scalebox{.4}{-$f$}};

\node[draw=none][blue] at (0.7, -0.75){\scalebox{.4}{$f$}};
\node[draw=none][blue] at (0.3, -0.3){\scalebox{.4}{$f$}};

\node[draw=none][gray] at (-0.7, -0.75){\scalebox{.4}{-$f$}};
\node[draw=none][gray] at (-0.3, -0.3){\scalebox{.4}{-$f$}};

\draw[edge][gray] (0) to (12);
\draw[edge][gray] (0) to (13);
\draw[edge][gray] (0) to (14);
\draw[edge][gray] (0) to (15);
\draw[edge][blue] (0) to (16);
\draw[edge][blue] (0) to (17);
\draw[edge][blue] (0) to (18);
\draw[edge][blue] (0) to (19);
\draw[edge][blue] (1) to (20);
\draw[edge][blue] (1) to (21);
\draw[edge][blue] (1) to (22);
\draw[edge][blue] (1) to (23);
\draw[edge][gray] (1) to (24);
\draw[edge][gray] (1) to (25);
\draw[edge][gray] (1) to (26);
\draw[edge][gray] (1) to (27);
\draw[edge][gray] (2) to (28);
\draw[edge][gray] (2) to (29);
\draw[edge][gray] (2) to (30);
\draw[edge][gray] (2) to (31);
\draw[edge][blue] (2) to (32);
\draw[edge][blue] (2) to (33);
\draw[edge][blue] (2) to (34);
\draw[edge][blue] (2) to (35);
\draw[edge][blue] (3) to (36);
\draw[edge][blue] (3) to (37);
\draw[edge][blue] (3) to (38);
\draw[edge][blue] (3) to (39);
\draw[edge][gray] (3) to (40);
\draw[edge][gray] (3) to (41);
\draw[edge][gray] (3) to (42);
\draw[edge][gray] (3) to (43);
\draw[edge][gray] (4) to (12);
\draw[edge][gray] (4) to (13);
\draw[edge][gray] (4) to (28);
\draw[edge][gray] (4) to (29);
\draw[edge][gray] (5) to (14);
\draw[edge][gray] (5) to (15);
\draw[edge][gray] (5) to (30);
\draw[edge][gray] (5) to (31);
\draw[edge][blue] (6) to (16);
\draw[edge][blue] (6) to (17);
\draw[edge][blue] (6) to (36);
\draw[edge][blue] (6) to (37);
\draw[edge][blue] (7) to (18);
\draw[edge][blue] (7) to (19);
\draw[edge][blue] (7) to (38);
\draw[edge][blue] (7) to (39);
\draw[edge][blue] (8) to (20);
\draw[edge][blue] (8) to (21);
\draw[edge][blue] (8) to (32);
\draw[edge][blue] (8) to (33);
\draw[edge][blue] (9) to (22);
\draw[edge][blue] (9) to (23);
\draw[edge][blue] (9) to (34);
\draw[edge][blue] (9) to (35);
\draw[edge][gray] (10) to (24);
\draw[edge][gray] (10) to (25);
\draw[edge][gray] (10) to (40);
\draw[edge][gray] (10) to (41);
\draw[edge][gray] (11) to (26);
\draw[edge][gray] (11) to (27);
\draw[edge][gray] (11) to (42);
\draw[edge][gray] (11) to (43);
\end{tikzpicture}

%% file: PlautLangLevel2From0to1Main.tikz.tex
\begin{tikzpicture}

\tikzset{vertex/.style={circle, draw, fill=black!50,inner sep=0pt, minimum width=2.2pt}} 
\tikzset{edge/.style = {draw=black, thick}}

\node[draw=none] at (6,1.6){\scalebox{1}{\scalebox{1.5}{$G_1$}}};



\node[draw=none] at (0,1.6){\scalebox{1}{$D_{4,12} \times {w_2}$}};

\node[draw=none] at (0,-1.6){\scalebox{1}{$D_{4,12} \times {w_1}$}};

\node[draw=none] at (-1.8,0.9){\scalebox{1}{$B_1$}};

\node[draw=none] at (1.8,0.9){\scalebox{1}{$B_1$}};

\draw[thick] [decorate,decoration= {brace,mirror,amplitude=0.2cm} ] (-1,-1) - - (1,-1);

\draw[thick] [decorate,decoration= {brace,amplitude=0.2cm} ] (-1,1) - - (1,1);

\draw[thick] [decorate,decoration= {brace,amplitude=0.2cm} ] (-2.7,0.3) - - (-1,0.3);

\draw[thick] [decorate,decoration= {brace,amplitude=0.2cm} ] (1,0.3) - - (2.7,0.3);

\node[vertex] (0g) at  (5.4,0.0) {\scalebox{.4}{}};
\node[vertex] (1g) at  (10.6,0.0) {\scalebox{.4}{}};
\node[vertex] (2g) at  (7.0,0.0) {\scalebox{.4}{}};
\node[vertex] (3g) at  (9.0,0.0) {\scalebox{.4}{}};
\node[vertex] (4g) at  (8.0,-1.3856406460551018) {\scalebox{.4}{}};
\node[vertex] (5g) at  (8.0,1.3856406460551018) {\scalebox{.4}{}};
\node[vertex] (6g) at  (5.8,0.0) {\scalebox{.4}{}};
\node[vertex] (7g) at  (6.6,0.0) {\scalebox{.4}{}};

\node[vertex] (9g) at  (6.2,0) {\scalebox{.4}{}};
\node[vertex] (10g) at  (10.2,0.0) {\scalebox{.4}{}};
\node[vertex] (11g) at  (9.4,0.0) {\scalebox{.4}{}};

\node[vertex] (13g) at  (9.8,0) {\scalebox{.4}{}};
\node[vertex] (14g) at  (7.25,-0.3473) {\scalebox{.4}{}};
\node[vertex] (15g) at  (7.75,-1.039) {\scalebox{.4}{}};
\node[vertex] (16g) at  (7.5,-0.69) {\scalebox{.4}{}};

\node[vertex] (18g) at  (7.25,0.3473) {\scalebox{.4}{}};
\node[vertex] (19g) at  (7.75,1.039) {\scalebox{.4}{}};

\node[vertex] (21g) at  (7.5,0.69) {\scalebox{.4}{}};
\node[vertex] (22g) at  (8.75,-0.3473) {\scalebox{.4}{}};
\node[vertex] (23g) at  (8.25,-1.039) {\scalebox{.4}{}};
\node[vertex] (24g) at  (8.5,-0.69) {\scalebox{.4}{}};

\node[vertex] (26g) at  (8.75,0.3473) {\scalebox{.4}{}};
\node[vertex] (27g) at  (8.25,1.039) {\scalebox{.4}{}};

\node[vertex] (29g) at  (8.5,0.69) {\scalebox{.4}{}};
\draw[edge] (0g) to (6g);
\draw[edge] (1g) to (10g);
\draw[edge] (2g) to (7g);
\draw[edge][gray] (2g) to (14g);
\draw[edge][blue] (2g) to (18g);
\draw[edge] (3g) to (11g);
\draw[edge][gray] (3g) to (22g);
\draw[edge][blue] (3g) to (26g);
\draw[edge][gray] (4g) to (15g);
\draw[edge][gray] (4g) to (23g);
\draw[edge][blue] (5g) to (19g);
\draw[edge][blue] (5g) to (27g);

\draw[edge] (6g) to (9g);

\draw[edge] (7g) to (9g);

\draw[edge] (10g) to (13g);

\draw[edge] (11g) to (13g);
\draw[edge][gray] (14g) to (16g);

\draw[edge][gray] (15g) to (16g);

\draw[edge][blue] (18g) to (21g);

\draw[edge][blue] (19g) to (21g);
\draw[edge][gray] (22g) to (24g);

\draw[edge][gray] (23g) to (24g);

\draw[edge][blue] (26g) to (29g);

\draw[edge][blue] (27g) to (29g);


\node[vertex] (0) at  (-2.6,0.0) {\scalebox{.4}{}};
\node[vertex] (1) at  (2.6,0.0) {\scalebox{.4}{}};
\node[vertex] (2) at  (-1.0,0.0) {\scalebox{.4}{}};

\node[vertex] (2a) at  (-1.0,0.69) {\scalebox{.4}{}};
\node[vertex] (2b) at  (-1.0,-0.69) {\scalebox{.4}{}};

\node[vertex] (3) at  (1.0,0.0) {\scalebox{.4}{}};

\node[vertex] (3a) at  (1.0,0.69) {\scalebox{.4}{}};
\node[vertex] (3b) at  (1.0,-0.69) {\scalebox{.4}{}};

\node[vertex] (4) at  (0.0,-0.69) {\scalebox{.4}{}};
\node[vertex] (5) at  (0.0,0.69) {\scalebox{.4}{}};
\node[vertex] (6) at  (-2.2,0.0) {\scalebox{.4}{}};
\node[vertex] (7) at  (-1.4,0.0) {\scalebox{.4}{}};

\node[vertex] (9) at  (-1.8,0) {\scalebox{.4}{}};
\node[vertex] (10) at  (2.2,0.0) {\scalebox{.4}{}};
\node[vertex] (11) at  (1.4,0.0) {\scalebox{.4}{}};

\node[vertex] (13) at  (1.8,0) {\scalebox{.4}{}};
\node[vertex] (14) at  (-0.75,-0.69) {\scalebox{.4}{}};
\node[vertex] (15) at  (-0.25,-0.69) {\scalebox{.4}{}};
\node[vertex] (16) at  (-0.5,-0.69) {\scalebox{.4}{}};

\node[vertex] (18) at  (-0.75,0.69) {\scalebox{.4}{}};
\node[vertex] (19) at  (-0.25,0.69) {\scalebox{.4}{}};

\node[vertex] (21) at  (-0.5,0.69) {\scalebox{.4}{}};
\node[vertex] (22) at  (0.75,-0.69) {\scalebox{.4}{}};
\node[vertex] (23) at  (0.25,-0.69) {\scalebox{.4}{}};
\node[vertex] (24) at  (0.5,-0.69) {\scalebox{.4}{}};

\node[vertex] (26) at  (0.75,0.69) {\scalebox{.4}{}};
\node[vertex] (27) at  (0.25,0.69) {\scalebox{.4}{}};

\node[vertex] (29) at  (0.5,0.69) {\scalebox{.4}{}};
\draw[edge] (0) to (6);
\draw[edge] (1) to (10);
\draw[edge] (2) to (7);
\draw[edge][gray] (2b) to (14);
\draw[edge][blue] (2a) to (18);
\draw[edge] (3) to (11);
\draw[edge][gray] (3b) to (22);
\draw[edge][blue] (3a) to (26);
\draw[edge][gray] (4) to (15);
\draw[edge][gray] (4) to (23);
\draw[edge][blue] (5) to (19);
\draw[edge][blue] (5) to (27);

\draw[edge] (6) to (9);

\draw[edge] (7) to (9);

\draw[edge] (10) to (13);

\draw[edge] (11) to (13);
\draw[edge][gray] (14) to (16);

\draw[edge][gray] (15) to (16);

\draw[edge][blue] (18) to (21);

\draw[edge][blue] (19) to (21);
\draw[edge][gray] (22) to (24);

\draw[edge][gray] (23) to (24);

\draw[edge][blue] (26) to (29);

\draw[edge][blue] (27) to (29);

\end{tikzpicture}


%% file: PlautLangLevel2From1to2Main.tikz.tex
\begin{tikzpicture}
\tikzset{vertex/.style={circle, draw, fill=black!50,inner sep=0pt, minimum width=2.2pt}}  
\tikzset{edge/.style = {draw=black, thick}}



\node[draw=none] at (4.3,1.3){\scalebox{1}{\scalebox{1.4}{$G_2=LP_2$}}};

\node[draw=none] at (-2,-0.9){\scalebox{1}{$\jacobi_{2,2}$}};

\draw[thick] [decorate,decoration= {brace,mirror,amplitude=0.1cm} ] (-2.4,-0.5) - - (-1.6,-0.5);


\draw[thick] [decorate,decoration= {brace,amplitude=0.1cm} ] (-1.1,0.5) - - (-0.7,1.05);

\node[draw=none] at (-1.4,1){\scalebox{1}{$\jacobi_{2,3}$}};

\node[vertex] (0) at  (-3.0,0.0) {\scalebox{.4}{}};
\node[vertex] (1) at  (3.0,0.0) {\scalebox{.4}{}};
\node[vertex] (2) at  (-1.0,0.0) {\scalebox{.4}{}};
\node[vertex] (3) at  (1.0,0.0) {\scalebox{.4}{}};
\node[vertex] (4) at  (0.0,-1.3856406460551018) {\scalebox{.4}{}};
\node[vertex] (5) at  (0.0,1.3856406460551018) {\scalebox{.4}{}};
\node[vertex] (6) at  (-2.3333333333333335,0.0) {\scalebox{.4}{}};
\node[vertex] (6a) at  (-2.3333333333333335,0.2) {\scalebox{.4}{}};
\node[vertex] (6b) at  (-2.3333333333333335,-0.2) {\scalebox{.4}{}};

\node[vertex] (7) at  (-1.6666666666666665,0.0) {\scalebox{.4}{}};
\node[vertex] (7a) at  (-1.6666666666666665,0.2) {\scalebox{.4}{}};
\node[vertex] (7b) at  (-1.6666666666666665,-0.2) {\scalebox{.4}{}};

\node[vertex] (89a) at  (-2.0,0.2) {\scalebox{.4}{}};
\node[vertex] (89b) at  (-2.0,-0.2) {\scalebox{.4}{}};

\node[vertex] (10) at  (2.3333333333333335,0.0) {\scalebox{.4}{}};
\node[vertex] (10a) at  (2.3333333333333335,0.2) {\scalebox{.4}{}};
\node[vertex] (10b) at  (2.3333333333333335,-0.2) {\scalebox{.4}{}};

\node[vertex] (11) at  (1.6666666666666665,0.0) {\scalebox{.4}{}};
\node[vertex] (11a) at  (1.6666666666666665,0.2) {\scalebox{.4}{}};
\node[vertex] (11b) at  (1.6666666666666665,-0.2) {\scalebox{.4}{}};

\node[vertex] (1213a) at  (2.0,0.2) {\scalebox{.39}{}};
\node[vertex] (1213b) at  (2.0,-0.2) {\scalebox{.39}{}};

\node[vertex] (14) at  (-0.6666666666666666,-0.4618802153517006) {\scalebox{.4}{}};
\node[vertex] (14a) at  (-0.4666666666666666,-0.4618802153517006) {\scalebox{.4}{}};
\node[vertex] (14b) at  (-0.8666666666666666,-0.4618802153517006) {\scalebox{.4}{}};
\node[vertex] (15) at  (-0.3333333333333333,-0.9237604307034012) {\scalebox{.4}{}};
\node[vertex] (15a) at  (-0.1333333333333333,-0.9237604307034012) {\scalebox{.4}{}};
\node[vertex] (15b) at  (-0.5333333333333333,-0.9237604307034012) {\scalebox{.4}{}};
\node[vertex] (1617a) at  (-0.3,-0.692820323027551) {\scalebox{.4}{}};
\node[vertex] (1617b) at  (-0.7,-0.692820323027551) {\scalebox{.4}{}};
\node[vertex] (18) at  (-0.6666666666666666,0.4618802153517006) {\scalebox{.4}{}};
\node[vertex] (18a) at  (-0.4666666666666666,0.4618802153517006) {\scalebox{.4}{}};
\node[vertex] (18b) at  (-0.8666666666666666,0.4618802153517006) {\scalebox{.4}{}};
\node[vertex] (19) at  (-0.3333333333333333,0.9237604307034012) {\scalebox{.4}{}};
\node[vertex] (19a) at  (-0.1333333333333333,0.9237604307034012) {\scalebox{.4}{}};
\node[vertex] (19b) at  (-0.5333333333333333,0.9237604307034012) {\scalebox{.4}{}};

\node[vertex] (2021a) at  (-0.3,0.692820323027551) {\scalebox{.4}{}};
\node[vertex] (2021b) at  (-0.7,0.692820323027551) {\scalebox{.4}{}};

\node[vertex] (22) at  (0.6666666666666666,-0.4618802153517006) {\scalebox{.39}{}};
\node[vertex] (22a) at  (0.4666666666666666,-0.4618802153517006) {\scalebox{.39}{}};
\node[vertex] (22b) at  (0.8666666666666666,-0.4618802153517006) {\scalebox{.39}{}};
\node[vertex] (23) at  (0.3333333333333333,-0.9237604307034012) {\scalebox{.4}{}};
\node[vertex] (23a) at  (0.1333333333333333,-0.9237604307034012) {\scalebox{.4}{}};
\node[vertex] (23b) at  (0.5333333333333333,-0.9237604307034012) {\scalebox{.4}{}};

\node[vertex] (2425a) at  (0.3,-0.692820323027551) {\scalebox{.4}{}};
\node[vertex] (2425b) at  (0.7,-0.692820323027551) {\scalebox{.4}{}};

\node[vertex] (26) at  (0.6666666666666666,0.4618802153517006) {\scalebox{.39}{}};
\node[vertex] (26a) at  (0.4666666666666666,0.4618802153517006) {\scalebox{.39}{}};
\node[vertex] (26b) at  (0.8666666666666666,0.4618802153517006) {\scalebox{.39}{}};
\node[vertex] (27) at  (0.3333333333333333,0.9237604307034012) {\scalebox{.4}{}};
\node[vertex] (27a) at  (0.1333333333333333,0.9237604307034012) {\scalebox{.4}{}};
\node[vertex] (27b) at  (0.5333333333333333,0.9237604307034012) {\scalebox{.4}{}};
\node[vertex] (2829a) at  (0.3,0.692820323027551) {\scalebox{.39}{}};
\node[vertex] (2829b) at  (0.7,0.692820323027551) {\scalebox{.39}{}};
\draw[edge] (0) to (6);
\draw[edge] (1) to (10);
\draw[edge] (2) to (7);
\draw[edge] (2) to (14);
\draw[edge] (2) to (18);
\draw[edge] (3) to (11);
\draw[edge] (3) to (22);
\draw[edge] (3) to (26);
\draw[edge] (4) to (15);
\draw[edge] (4) to (23);
\draw[edge] (5) to (19);
\draw[edge] (5) to (27);
\draw[edge][blue] (6a) to (89a);
\draw[edge][gray] (6b) to (89b);
\draw[edge][blue] (7a) to (89a);
\draw[edge][gray] (7b) to (89b);
\draw[edge][blue] (10a) to (1213a);
\draw[edge][gray] (10b) to (1213b);
\draw[edge][blue] (11a) to (1213a);
\draw[edge][gray] (11b) to (1213b);

\draw[edge][blue] (14a) to (1617a);
\draw[edge][blue] (15a) to (1617a);
\draw[edge][gray] (14b) to (1617b);
\draw[edge][gray] (15b) to (1617b);

\draw[edge][blue] (18a) to (2021a);
\draw[edge][gray] (18b) to (2021b);
\draw[edge][blue] (19a) to (2021a);
\draw[edge][gray] (19b) to (2021b);
\draw[edge][blue] (22a) to (2425a);
\draw[edge][blue] (23a) to (2425a);
\draw[edge][gray] (22b) to (2425b);
\draw[edge][gray] (23b) to (2425b);
\draw[edge][blue] (26a) to (2829a);
\draw[edge][gray] (26b) to (2829b);
\draw[edge][blue] (27a) to (2829a);
\draw[edge][gray] (27b) to (2829b);


\node[vertex] (0g) at  (4.0,0.0) {\scalebox{.4}{}};
\node[vertex] (1g) at  (10.0,0.0) {\scalebox{.4}{}};
\node[vertex] (2g) at  (6.0,0.0) {\scalebox{.4}{}};
\node[vertex] (3g) at  (8.0,0.0) {\scalebox{.4}{}};
\node[vertex] (4g) at  (7.0,-1.3856406460551018) {\scalebox{.4}{}};
\node[vertex] (5g) at  (7.0,1.3856406460551018) {\scalebox{.4}{}};
\node[vertex] (6g) at  (4.6666666666,0.0) {\scalebox{.4}{}};
\node[vertex] (7g) at  (5.333333333,0.0) {\scalebox{.4}{}};
\node[vertex] (8g) at  (5.0,-0.4618802153517008) {\scalebox{.4}{}};
\node[vertex] (9g) at  (5.0,0.4618802153517008) {\scalebox{.4}{}};
\node[vertex] (10g) at  (9.33333333333,0.0) {\scalebox{.4}{}};
\node[vertex] (11g) at  (8.6666666666666665,0.0) {\scalebox{.4}{}};
\node[vertex] (12g) at  (9.0,-0.4618802153517008) {\scalebox{.39}{}};
\node[vertex] (13g) at  (9.0,0.4618802153517008) {\scalebox{.39}{}};
\node[vertex] (14g) at  (6.3333333333,-0.4618802153517006) {\scalebox{.4}{}};
\node[vertex] (15g) at  (6.66666666666,-0.9237604307034012) {\scalebox{.4}{}};
\node[vertex] (16g) at  (6.18,-0.9237604307034013) {\scalebox{.4}{}};
\node[vertex] (17g) at  (6.82,-0.4618802153517006) {\scalebox{.4}{}};
\node[vertex] (18g) at  (6.333333333,0.4618802153517006) {\scalebox{.4}{}};
\node[vertex] (19g) at  (6.6666666666,0.9237604307034012) {\scalebox{.4}{}};
\node[vertex] (20g) at  (6.82,0.4618802153517006) {\scalebox{.4}{}};
\node[vertex] (21g) at  (6.18,0.9237604307034013) {\scalebox{.4}{}};
\node[vertex] (22g) at  (7.6666666666666666,-0.4618802153517006) {\scalebox{.39}{}};
\node[vertex] (23g) at  (7.3333333333333333,-0.9237604307034012) {\scalebox{.4}{}};
\node[vertex] (24g) at  (7.82,-0.9237604307034013) {\scalebox{.39}{}};
\node[vertex] (25g) at  (7.18000000000000005,-0.4618802153517006) {\scalebox{.4}{}};
\node[vertex] (26g) at  (7.6666666666666666,0.4618802153517006) {\scalebox{.39}{}};
\node[vertex] (27g) at  (7.3333333333333333,0.9237604307034012) {\scalebox{.4}{}};
\node[vertex] (28g) at  (7.18000000000000005,0.4618802153517006) {\scalebox{.4}{}};
\node[vertex] (29g) at  (7.82,0.9237604307034013) {\scalebox{.39}{}};
\draw[edge] (0g) to (6g);
\draw[edge] (1g) to (10g);
\draw[edge] (2g) to (7g);
\draw[edge] (2g) to (14g);
\draw[edge] (2g) to (18g);
\draw[edge] (3g) to (11g);
\draw[edge] (3g) to (22g);
\draw[edge] (3g) to (26g);
\draw[edge] (4g) to (15g);
\draw[edge] (4g) to (23g);
\draw[edge] (5g) to (19g);
\draw[edge] (5g) to (27g);
\draw[edge][gray] (6g) to (8g);
\draw[edge][blue] (6g) to (9g);
\draw[edge][gray] (7g) to (8g);
\draw[edge][blue] (7g) to (9g);
\draw[edge][gray] (10g) to (12g);
\draw[edge][blue] (10g) to (13g);
\draw[edge][gray] (11g) to (12g);
\draw[edge][blue] (11g) to (13g);
\draw[edge][gray] (14g) to (16g);
\draw[edge][blue] (14g) to (17g);
\draw[edge][gray] (15g) to (16g);
\draw[edge][blue] (15g) to (17g);
\draw[edge][blue] (18g) to (20g);
\draw[edge][gray] (18g) to (21g);
\draw[edge][blue] (19g) to (20g);
\draw[edge][gray] (19g) to (21g);
\draw[edge][gray] (22g) to (24g);
\draw[edge][blue] (22g) to (25g);
\draw[edge][gray] (23g) to (24g);
\draw[edge][blue] (23g) to (25g);
\draw[edge][blue] (26g) to (28g);
\draw[edge][gray] (26g) to (29g);
\draw[edge][blue] (27g) to (28g);
\draw[edge][gray] (27g) to (29g);
\end{tikzpicture}